\documentclass{article}
\usepackage{amsmath,amsthm,amssymb}
\newtheorem{Theorem}{Theorem}[section]
\newtheorem{lem}[Theorem]{Lemma}
\newtheorem{Remark}[Theorem]{Remark}
\newtheorem{Definition}[Theorem]{Definition}
\newtheorem{Corollary}[Theorem]{Corollary}

\numberwithin{equation}{section}
\usepackage[latin1]{inputenc}%%%%½ÓÊÜ²»Í¬µÄÊäÈë±àÂëºê°ü; Í¶arxiv±ØÐëÓÐÕâ¸öÃüÁî
\begin{document}
\title{{\LARGE The number of  extended irreducible binary Goppa codes}}
%\title{{\LARGE Enumeration of extended irreducible binary Goppa codes for arbitrary degree}}
%%%\author{Author}%Guanghui Zhang Hongwei Liu
%\thanks{G.Zhang is with School of Mathematical Sciences,Luoyang Normal University,
%Luoyang, Henan, 471934, China}
%}

\author{Bocong Chen$^1$ and Guanghui~Zhang$^2$\footnote{E-mail addresses: bocongchen@foxmail.com(B. Chen),~zghui@lynu.edu.cn (G. Zhang)}} %and~Jane~Doe,~\IEEEmembership{Life~Fellow,~IEEE}% <-this % stops a space

%%\thanks{This work was supported
%%in part by NSFC under Grant No.11571005,
%%in part by the Science and Technology Development Program of Henan Province in 2019 under Grant No.192102310444,
%%and in part by the Key Research Project of Higher Education of the Education Department of Henan Province in 2019 under Grant No. 19A120010.
%%H. Liu was partially supported by NSFC under Grant No.11871025.
%%}
%%%\thanks{G. Zhang is with the School of Mathematical Sciences, Luoyang Normal University,
%%%Luoyang, Henan 471934, China (email: zghui@lynu.edu.cn).
%%%B. Chen is with the School of Mathematics, South China University of Technology, Guangzhou 51064, China
%%%(email: bocong\underline{~~}chen@yahoo.com).}
%%\thanks{H. Liu is with the School of Mathematics and Statistics,
%%Central China Normal University, Wuhan 430079, China (email: hwliu@mail.ccnu.edu.cn).}% <-this % stops a space
%}

\date{\small
$1.$ School of Mathematics, South China University of Technology, Guangzhou 510641, China\\
$2.$ School of Mathematical Sciences, Luoyang Normal University, Luoyang, Henan 471934, China
}

%%\date{}%%²»ÏÔÊ¾ÎÄÕÂÈÕÆÚµÄÃüÁî

\maketitle

\begin{abstract}
Goppa, in the 1970s, discovered the relation between
algebraic geometry and codes, which led to the family of Goppa codes.
As  one of the most interesting subclasses of linear codes,
the family of Goppa codes is often chosen as a key in the McEliece cryptosystem.
Knowledge of the number of inequivalent binary Goppa codes for fixed parameters may facilitate in the evaluation
of the security of such a cryptosystem.
Let $n\geq5$ be an odd prime number,
let $q=2^n$ and let $r\geq3$ be a positive integer satisfying $\gcd(r,n)=1$.
The purpose of this paper is to establish an upper bound on the number of inequivalent extended irreducible binary Goppa codes
of length $q+1$ and degree $r$.
A potential mathematical object for this purpose is to count the number of orbits of
  the projective semi-linear group ${\rm PGL}_2(\mathbb{F}_q)\rtimes{\rm Gal}(\mathbb{F}_{q^r}/\mathbb{F}_2)$
on the set $\mathcal{I}_r$ of all monic irreducible polynomials of degree $r$ over the finite field $\mathbb{F}_q$.
An explicit  formula for the number of orbits of ${\rm PGL}_2(\mathbb{F}_q)\rtimes{\rm Gal}(\mathbb{F}_{q^r}/\mathbb{F}_2)$
on  $\mathcal{I}_r$ is given, and consequently,  an upper bound for the number of inequivalent extended irreducible binary Goppa codes
of length $q+1$ and degree $r$ is derived. Our main result naturally contains the main results of
Ryan (IEEE-TIT 2015), Huang and Yue (IEEE-TIT, 2022) and, Chen and Zhang (IEEE-TIT, 2022), which considered the cases
$r=4$, $r=6$ and $\gcd(r,q^3-q)=1$ respectively.

\medskip
\textbf{MSC:} 94B50.

\textbf{Keywords:} Binary Goppa codes, extended Goppa codes, inequivalent codes, group actions.

\end{abstract}

\section{Introduction}
The progress of cryptography is closely related with the development of coding theory.
Analogous to the RSA, coding theory  started shaping public key cryptography in the late 1970s.  McEliece
introduced a public-key cryptosystem based upon encoding the plaintext as codewords of an error correcting
code from the family of Goppa codes in 1978, see \cite{mceliece}.
In the McEliece public-key encryption scheme the main practical limitation is probably the size of its key.
In order to overcome this practical limitation, the McEliece cryptosystem often chooses a random Goppa code  as its key, see \cite{ls}, \cite{mceliece}.
In the originally proposed system, a codeword is generated from
plaintext message bits by using a permuted and scrambled generator matrix of a Goppa code.
This matrix is the public key.
In this system, the ciphertext is formed by adding
a randomly chosen error vector  to each codeword of perturbed
code. The unperturbed Goppa code, together with scrambler and permutation matrices,  form the private key. On
reception, the associated private key is used to invoke an error-correcting decoder based upon the underlying
Goppa code to correct the garbled bits in the codeword.

One of the  reasons why Goppa codes receive interest from cryptographers may be
that  Goppa codes have few invariants and the number of inequivalent codes grow exponentially
with the length and dimension of the code, which makes it possible to resist to any structural attack.
When we give the assessment of the security of this cryptosystem against the enumerative attack,
it is important for us to know the number of Goppa codes for any given set of parameters.
An enumerative attack in the McEliece cryptosystem is to find all Goppa codes for a given set of parameters and
to test their equivalences with the public codes \cite{ls}.
Thus one of the key issues for the McEliece cryptosystem is the enumeration of inequivalent Goppa codes for a given set of parameters.
Knowledge of the number of inequivalent Goppa codes for fixed parameters may facilitate in the evaluation
of the security of such a cryptosystem.

\textbf{A. Known results}

Some significant research efforts have been put  in developing the enumeration of  (extended) Goppa codes.
Based on the invariant property under the group of transformations, Moreno \cite{moreno} classified cubic and quartic irreducible
Goppa codes; in the same paper, it was showed that there are four inequivalent  quartic Goppa codes of length $33$ and
there is only one inequivalent extended irreducible binary Goppa code with any length and degree $3$.
Berger \cite{berger3,berger1} studied Goppa codes that are invariant under  a prescribed permutation.
Ryan and Fitzpatrick \cite{rf} obtained an upper bound for the number of inequivalent irreducible Goppa codes of length $q^n$ over $\mathbb{F}_q$.
Ryan \cite{ryan14} produced an upper bound on the number of inequivalent extended irreducible
Goppa codes over $\mathbb{F}_q$ of degree $r$ and length $q^n+1$.

In a subsequent paper \cite{ryan15},
Ryan made a great improvement on  giving a much tighter upper bound than that of \cite{ryan14}
on the number of  inequivalent extended irreducible binary quartic Goppa codes of
length $2^n+1$, where $n>3$ is a prime number. It was shown in \cite{ryan15} that the problem of giving an upper bound for the number of
inequivalent extended irreducible binary Goppa codes of degree $r$ can be
transformed into that of finding the number of orbits of the   projective semi-linear group ${\rm P\Gamma L}={\rm PGL}_2(\mathbb{F}_q)\rtimes{\rm Gal}(\mathbb{F}_{q^r}/\mathbb{F}_2)$
on the set of   elements in $\mathbb{F}_{q^r}$ of degree $r$ over $\mathbb{F}_q$ (which is denoted by $\mathcal{S}$).
The objective of the paper \cite{ryan15} is then to find such number of orbits.
Following that line of research, Musukwa {\it et al.} \cite{mmr} gave
an upper bound on the number of inequivalent extended irreducible
binary Goppa codes of degree $2^m$ and length $2^n+1$, where $n$ is an odd prime and $m>1$ is a positive integer.
Musukwa produced \cite{musukwa} an upper bound on the number of inequivalent extended irreducible binary Goppa codes of degree $2p$ and length
$2^n+1$, where $n$ and $p$ are two distinct odd primes such that $p$ does not divide $2^n\pm1$.
Magamba and Ryan \cite{magamba} obtained an upper bound on the number of inequivalent extended irreducible $q$-ary Goppa codes
of degree $r$ and length $q^n+1$, where $q=p^t$, $n$ and $r>2$ are both prime numbers.
Recently, Huang and Yue \cite{yueqin} obtained an upper bound on the number of extended irreducible binary Goppa codes of degree $6$ and length
$2^n+1$, where $n>3$ is a prime number.
Note that the degrees of the Goppa codes mentioned above are small or have at most two prime divisors.
Chen and Zhang \cite{cz} presented a new approach to
calculate the number of orbits of the   projective semi-linear group
on $\mathcal{S}$ yielding an upper bound on the number of extended irreducible binary Goppa codes of degree $r$ and length
$2^n+1$, where $n>3$ is a prime number with $\gcd(r,n)=1$ and $\gcd(r,q^3-q)=1$.
In particular, the degree $r$ of the Goppa code considered in \cite{cz} can have arbitrary many prime divisors.

\textbf{B. Our main results and contributions}

In this paper,  we further explore the ideas in \cite{cz}
to establish an upper bound on the number of inequivalent extended irreducible binary Goppa codes of length $q+1$ and degree $r$, where
$q=2^{n}$ and
$n\geq5$ is a prime number satisfying $\gcd(r,n)=1$.
In a word, we settle a much more general case by dropping  the assumption
$\gcd(r,q^3-q)=1$ in \cite{cz}; consequently, our main results in the current paper naturally contain the main results of
\cite{cz}, \cite{yueqin},  \cite{musukwa}, \cite{mmr} and \cite{ryan15}.
A potential mathematical object for this purpose is to count the number of orbits of
${\rm P\Gamma L}={\rm PGL}_2(\mathbb{F}_q)\rtimes{\rm Gal}(\mathbb{F}_{q^r}/\mathbb{F}_2)$
on $\mathcal{S}$ (see Lemma \ref{important} in Section 2).
We first use a   strategy exhibited in \cite{cz} to count the number of orbits of ${\rm P\Gamma L}_2(\mathbb{F}_q)$ on $\mathcal{I}_r$, where
 $\mathcal{I}_r$ denotes the set of monic irreducible polynomials over $\mathbb{F}_q$ of degree $r$ (see Lemmas \ref{action} and \ref{orbit} in Section 2).
By virtue of a result in \cite{ryan15}, the number of inequivalent extended irreducible binary Goppa codes of length $q+1$
and degree $r$ is less than or equal to the number of orbits of ${\rm P\Gamma L}_2(\mathbb{F}_q)$ on $\mathcal{S}$.
We finally determine the exact value of the number of orbits of ${\rm P\Gamma L}$ on $\mathcal{I}_r$ (or equivalently ${\rm P\Gamma L}$ on $\mathcal{S}$), see Theorem \ref{theorem} in Section 4.
Comparing  to \cite{cz},
without the assumption $\gcd(r,q^3-q)=1$,
we have to get around several difficulties in connecting the orbits of
${\rm P\Gamma L}$ on $\mathcal{I}_r$ and that on $\mathcal{S}$ (see Lemmas \ref{equivalent}-\ref{lastlem} in Section 4) and establish some new results
(see Lemmas \ref{cyclic}-\ref{delta5}).
The auxiliary results may be interested in their own right.

\textbf{C. Organization of this paper}

The paper is organized as follows.
In Section $2$, we review some definitions and basic results about extended irreducible Goppa codes, some matrix groups and group actions.
In Section  $3$,  we study the   number of orbits of ${\rm P\Gamma L}$ on $\mathcal{I}_r$.
In Section  $4$, we find an explicit formula for the number of orbits of ${P \Gamma L}$ on the set $\mathcal{I}_r$, which naturally gives an upper bound for the number of inequivalent extended irreducible Goppa codes of length $2^n+1$ and degree $r$, where $n\geq5$
is a prime number satisfying $\gcd(r,n)=1$.
In Section $5$, as corollaries of our main results, we apply our main result to some special cases, including $r=4$, $2p$ ($p\geq3$ is a prime number) and $\gcd(r,q^3-q)=1$.

\section{Preliminaries}
Starting  from this section till the end of this paper, we assume that  $n\geq 5$ is an odd prime number
and  $r\geq3$ is a positive integer relatively prime to $n$.
Let $\mathbb{F}_q$ be the finite field with $q=2^n$ elements and let $\mathbb{F}_q^*=\mathbb{F}_q\setminus\{0\}$
be the multiplicative group of the finite field $\mathbb{F}_q$.  Suppose $x$ is an indeterminate over
$\mathbb{F}_q$ and let $\mathbb{F}_q[x]$ be the polynomial ring in   variable $x$  with coefficients in $\mathbb{F}_q$.
As usual, for a polynomial $f(x)\in \mathbb{F}_q[x]$ (or simply denoted by $f$), $\deg f$ is the degree of $f$;
for a finite set $X$, let $|X|$  denote the number of elements of $X$. Given two integers $a$ and $b$, if $a$ is a divisor
of $b$, we write $a\mid b$; otherwise, we write $a\nmid b$. We use $\gcd(a,b)$ to denote the greatest common divisor of $a$ and $b$.
In particular, when $a$ and $b$ are relatively prime, we have $\gcd(a,b)=1$.

We begin with recalling the notion of irreducible binary Goppa codes of length $q$.
For the general definition and more detail information about Goppa codes,   readers may refer to \cite{lx} or \cite{ms}.
\subsection{Extended irreducible Goppa codes}
\begin{Definition}\label{definition1}
Let $g(x)$ be a polynomial in $\mathbb{F}_q[x]$ of degree $r$ and let
$L=\mathbb{F}_q=\{\alpha_0,\alpha_1,\cdots,\alpha_{q-1}\}$  satisfy
$g(\alpha_j)\neq0$ for any $0\leq j\leq q-1$.
The binary Goppa code $\Gamma(L,g)$ of length $q$ and degree $r$ is defined as
$$\Gamma(L,g)=\bigg\{c=(c_0,c_1,\cdots,c_{q-1})\in \mathbb{F}_2^q\,\Big|\,\sum_{i=0}^{q-1}\frac{c_i}{x-\alpha_i}\equiv 0\pmod{g(x)}\bigg\}.$$
The polynomial $g(x)$ is called the Goppa polynomial.
When $g(x)$ is irreducible, $\Gamma(L,g)$ is called an irreducible binary Goppa code of degree $r$.
\end{Definition}
%Note that $\frac{1}{x-\alpha_i}$ can be regarded as a polynomial modulo $g(x)$:
%$$\frac{1}{x-\alpha_i}=-\frac{g(x)-g(\alpha_i)}{x-\alpha_i}g(\alpha_i)^{-1}\pmod{g(x)}.$$
%It follows that $c\in \Gamma(L,g)$ if and only if
%$$\sum_{i=1}^mc_i\frac{g(x)-g(\alpha_i)}{x-\alpha_i}g(\alpha_i)^{-1}=0$$
%as a polynomial.
%Thus there exists an equivalent definition of binary  Goppa codes.
%\begin{Definition}
%For a given Goppa polynomial $g(x)$ of degree $r$ and $L=\{\alpha_1,\alpha_2,\cdots,\alpha_m\}$ with $g(\alpha_j)\neq0$ for any $1\leq j\leq m$, we have
%$$\Gamma(L,g)=\Big\{c\in \mathbb{F}_2^m\,\big|\,cH^T=0\Big\},$$
%where
%$$H=\begin{pmatrix}
%g(\alpha_1)^{-1} & g(\alpha_2)^{-1} & \cdots & g(\alpha_m)^{-1}\\
%\alpha_1g(\alpha_1)^{-1} & \alpha_2g(\alpha_2)^{-1} & \cdots & \alpha_mg(\alpha_m)^{-1}\\
%\vdots & \vdots & \ddots & \vdots\\
%\alpha_{1}^{r-1}g(\alpha_1)^{-1} & \alpha_{2}^{r-1}g(\alpha_2)^{-1} & \cdots & \alpha_m^{r-1}g(\alpha_m)^{-1}\\
%\end{pmatrix}.
%$$
%\end{Definition}
The Goppa code of length $q$ can be extended to a code of length $q+1$ by appending a coordinate in the set $L=\mathbb{F}_q$.
In this paper, we mainly consider extended irreducible binary Goppa codes.
The definition of  extended irreducible binary Goppa codes of length $q+1$ and degree $r$ is given below.

\begin{Definition}
For a given monic irreducible polynomial $g(x)$ of degree $r$, let $\Gamma(L,g)$ be an irreducible binary Goppa code
of length $q$ as given in Definition \ref{definition1}.
The extended Goppa code $\overline{\Gamma(L,g)}$ of length $q+1$ is defined as
$$\overline{\Gamma(L,g)}=\Big\{\big(c_0,c_1,\cdots,c_q\big)\in \mathbb{F}_2^{q+1}\,\Big|\,
\big(c_0,c_1,\cdots,c_{q-1}\big)\in \Gamma(L,g)~\hbox{and}~\sum\limits_{i=0}^qc_i=0\Big\}.$$
\end{Definition}
Chen \cite{Chen} showed that the irreducible binary Goppa code $\Gamma(L,g)$ is completely determined by any root of the Goppa polynomial $g(x)$;
more precisely, if $\alpha$ is a root of $g(x)$ in some extension field over $\mathbb{F}_q$, then
$$
H(\alpha)=\Big(\frac{1}{\alpha-\alpha_0},\frac{1}{\alpha-\alpha_1},\cdots,\frac{1}{\alpha-\alpha_{q-1}}\Big)
$$
can be served as a parity-check matrix for $\Gamma(L,g)$.
As such, let $C(\alpha)$ denote the code $\Gamma(L,g)$ and let  $\overline{C(\alpha)}$ denote the code $\overline{\Gamma(L,g)}$.
Therefore, every extended irreducible binary Goppa code of length $q+1$ and degree $r$ can be described as $\overline{C(\alpha)}$
for some $\alpha\in \mathbb{F}_{q^r}$.

\subsection{Equivalent extended irreducible Goppa codes}
%Of particular interest to us is the number of (inequivalent) extended irreducible Goppa codes.
%Berger \cite{berger3} showed that this number is closely related to the number of orbits of certain group actions.
In this paper, we aim to give an
upper bound for the number of inequivalent
extended irreducible binary   Goppa codes  of length $q+1$ and degree $r$.
%inequivalent codes in
%$$
%\Big\{\overline{C(\alpha)}\,\Big|\,\alpha\in \mathbb{F}_{q^r}\Big\}.
%$$
This problem can be reduced to that of  counting
the number of orbits of the projective semi-linear group action on some subset of $\mathbb{F}_{q^r}$ (see \cite{berger3}, \cite{yueqin} or \cite{ryan15}).
To state this result clearly, we need the notions of group actions (for example, see \cite{Kerber} or \cite{Rotman}) and some matrix groups.
In the following,  we collect the matrix groups that we will
use later,  and fix the notations.

(1) The general linear group of degree $2$ over $\mathbb{F}_q$
$${\rm GL}={\rm GL}_2(\mathbb{F}_q)=\bigg\{
\begin{pmatrix}
a & b \\
c & d
\end{pmatrix}\bigg|~a,b,c,d\in \mathbb{F}_q, ~ad-bc\neq 0\bigg\}.$$

(2) The affine general linear group of degree $2$ over $\mathbb{F}_q$
$${\rm AGL}={\rm AGL}_2(\mathbb{F}_q)=\bigg\{
\begin{pmatrix}
a & b \\
0 & 1
\end{pmatrix}\bigg{|}~a\in \mathbb{F}_q^*, ~b\in \mathbb{F}_q\bigg\}.$$

(3) The projective general linear group of degree $2$ over $\mathbb{F}_q$
$${\rm PGL}={\rm PGL}_2(\mathbb{F}_q)={\rm GL}/\mathcal{Z},$$
where $\mathcal{Z}$ is the center of ${\rm GL}$ consisting of the multiples of the identity matrix by elements of $\mathbb{F}_q^*$.

(4) The projective semi-linear group
$${\rm P\Gamma L}={\rm P\Gamma L}_2(\mathbb{F}_q)={\rm PGL}\rtimes {\rm Gal}=
\Big\{A\sigma^i\,\Big|\,A\in {\rm PGL}, ~0\leq i\leq rn-1\Big\},$$
where ${\rm Gal}={\rm Gal}(\mathbb{F}_{q^r}/\mathbb{F}_2)={\rm Gal}(\mathbb{F}_{2^{rn}}/\mathbb{F}_2)=\langle\sigma \rangle$ is
the Galois group of order $rn$ generated by $\sigma$ ($\sigma$ sends each $\alpha\in\mathbb{F}_{q^r}$ to $\alpha^2$).
The operation $``\cdot"$ in ${\rm P\Gamma L}$ is defined as follows:
$$A\sigma^i\cdot B\sigma^j=A\sigma^i(B)\sigma^{i+j}, ~ 0\leq i,j\leq rn-1,$$
where $\sigma^i(B)=\begin{pmatrix}\sigma^it & \sigma^iu\\ \sigma^iv & \sigma^iw\end{pmatrix}$ for
$B=\begin{pmatrix}t & u\\ v & w\end{pmatrix}\in {\rm PGL}$ ($\sigma^ia$ means $\sigma^ia=a^{2^i}$ for $a\in\mathbb{ F}_{q}$).
It is clear that $E_2\sigma^0$ is the identity element of $P\Gamma L$, where $E_2=\left(
                                                                                    \begin{array}{cc}
                                                                                      1 & 0 \\
                                                                                      0 & 1 \\
                                                                                    \end{array}
                                                                                  \right)
$
is the identity matrix.

Now it is the turn of group  actions.
For a general group $H$ acting on a finite set $X$, let $H(x)$ denote the orbit containing $x\in X$, namely $H(x)=\{hx\,|\,h\in H\}$; let
${\rm Stab_H}(x)$ be the stabilizer of the point $x\in X$ in $H$, namely  ${\rm Stab_H}(x)=\{h\in H\,|\, hx=x\}$.
Then the cardinality of the orbit $H(x)$ is equal to the index of ${\rm Stab_H}(x)$ in $H$ and is written
$$\big|H(x)\big|=\big[H:{\rm Stab_H}(x)\big].$$
%Let $\overline{\mathbb{F}}_q=\mathbb{F}_q\cup \{\infty\}$ be the projective line set.
Now let $\mathcal{S}=\mathcal{S}(r,n)$ denote the set of  elements in $\mathbb{F}_{q^r}$ of degree $r$ over $\mathbb{F}_q$; in other words,
$$\mathcal{S}=\Big\{\alpha\in \mathbb{F}_{q^r}\,\Big|\,\hbox{there exists a monic irreducible polynomial $f$ of degree $r$ over $\mathbb{F}_q$ satisfying $f(\alpha)=0$}\Big\}.$$
It is   known that ${\rm PGL}$ and ${\rm P\Gamma L}$ can act on the set $\mathcal{S}$ in the following ways (see \cite{yueqin} or \cite{ryan15}):
\begin{itemize}

\item The action of the projective general linear group on $\mathcal{S}$:
\begin{eqnarray*}
{\rm PGL}\times \mathcal{S} &\longrightarrow & \mathcal{S}\\
(A,~\alpha) &\mapsto & A\alpha=\frac{a\alpha+b}{c\alpha+d},
\end{eqnarray*}
where $A=\begin{pmatrix}a & b \\ c & d\end{pmatrix}\in {\rm PGL}$.
%$\frac{1}{0}=\infty$ and $\frac{1}{\infty}=0$.

\item The action of the projective semi-linear group on $\mathcal{S}$:
\begin{eqnarray*}
{\rm P\Gamma L}\times \mathcal{S} &\longrightarrow & \mathcal{S}\\
\Big( A\sigma^i,~\alpha\Big) &\mapsto & (A\sigma^i)\alpha=A\big(\sigma^i(\alpha)\big)=\frac{a\sigma^i(\alpha)+b}{c\sigma^i(\alpha)+d}
=\frac{a\alpha^{2^{i}}+b}{c\alpha^{2^{i}}+d}.
\end{eqnarray*}
\end{itemize}
We are ready to state a sufficient condition
which guarantees two extended irreducible Goppa
codes to be equivalent; thus, in particular, it gives an upper bound for the number of inequivalent codes in
$$
\Big\{\overline{C(\alpha)}\,\Big|\,\alpha\in \mathcal{S}\Big\},
$$
see \cite{berger3}, \cite{yueqin} or \cite{ryan15}.
\begin{lem}\label{important}
Let $\alpha\in \mathcal{S}$ and $\beta\in \mathcal{S}$. If $\alpha,\beta$ lie in the same ${\rm P\Gamma L}$-orbit,
namely $\alpha=A\sigma^i\beta$ for some
$A\sigma^i\in {\rm P\Gamma L}$, then the extended Goppa code $\overline{C(\alpha)}$ is (permutation) equivalent to the extended Goppa code
$\overline{C(\beta)}$.
In particular, the number of inequivalent extended irreducible binary Goppa codes of length $q+1$
and degree $r$ is less than or equal to the number of orbits of ${\rm P\Gamma L}$ on $\mathcal{S}$.
\end{lem}
With the help of Lemma \ref{important},   we only need to count the number of orbits of ${\rm P\Gamma L}$ on $\mathcal{S}$.

\subsection{The action of ${\rm P\Gamma L}$ on $\mathcal{I}_r$}
In this subsection,  we introduce another group action: The group $\rm{P \Gamma L}$ can act on the set of all
monic irreducible polynomials of degree $r$ over $\mathbb{F}_q$. Let $\mathcal{I}_r$
be the set of all monic irreducible polynomials of degree $r$ over  $\mathbb{F}_q$.
It has been shown that the number of orbits of $\rm{P \Gamma L}$ on $\mathcal{S}$ is equal to the number of orbits of
$\rm{P \Gamma L}$ on $\mathcal{I}_r$, see \cite{cz}.

Let $A=\begin{pmatrix}a & b \\ c & d\end{pmatrix}\in {\rm PGL}$, $\alpha \in \mathbb{F}_q$
and $f(x)=a_0+a_1x+\cdots+a_rx^r\in \mathbb{F}_q[x]$ with $a_r\neq 0$. We make the following definitions:
\begin{eqnarray*}
\big(f(x)\big)^*&=&\frac{1}{a_r}f(x),~~
A\alpha=\frac{a\alpha+b}{c\alpha+d},\\
Af&=&(-cx+a)^{r}f\big(A^{-1}x\big)=(-cx+a)^{r}f\Big(\frac{dx-b}{-cx+a}\Big),\\
\sigma^if&=&\sigma^i(f(x))=\sigma^i(a_0)+\sigma^i(a_1)x+\cdots+\sigma^i(a_r)x^r.
\end{eqnarray*}

The group $\rm{P\Gamma L}$ can act on the set $\mathcal{I}_r$, as restated below.
\begin{lem}\label{action}
{\rm (\cite[Lemma 3.1]{cz})} With  notation given above, we have a group action  ${\rm P\Gamma L}$ on the set $\mathcal{I}_r$ defined by
\begin{eqnarray*}
{\rm P\Gamma L}\times \mathcal{I}_r &\rightarrow & \mathcal{I}_r\\
\big(A\sigma^i,f\big) &\mapsto &(A\sigma^i)\big(f\big)=\Big(A(\sigma^if)\Big)^*.
\end{eqnarray*}
\end{lem}

\begin{Remark}{\rm
Many authors have studied the action of ${\rm PGL}$ on $\mathcal{I}_r$, focusing on the characterization
and number of $A$-invariants where $A\in{\rm PGL}$ (for example, see \cite{Gare}, \cite{Reis182}, \cite{Reis20}, \cite{Reis18}, \cite{ST}).
The paper \cite{MOR} considered an action of ${\rm P\Gamma L}$ on $\mathcal{I}_r$, and our definition of ${\rm P\Gamma L}$ on $\mathcal{I}_r$ is different from  that of \cite{MOR}}.
\end{Remark}

The next result reveals that the problem of counting the number of orbits of  ${\rm P\Gamma L}$ on $\mathcal{S}$ can be completely converted
to that of counting  the number of orbits of ${\rm P\Gamma L}$ on $\mathcal{I}_r$.
\begin{lem}\label{orbit}
{\rm (\cite[Lemma 3.3]{cz})} The number of orbits of  ${\rm P\Gamma L}$ on $\mathcal{S}$ is equal to the number of orbits of ${\rm P\Gamma L}$ on $\mathcal{I}_r$.
\end{lem}

By Lemma \ref{orbit},
our ultimate aim is to find the number of   orbits of  ${\rm P\Gamma L}={\rm PGL}\rtimes {\rm Gal}$ on the set $\mathcal{I}_r$.
We will repeatedly use the following fact to achieve this goal (for example, see \cite[Pages 35-36]{Kerber}):
\begin{lem}\label{book}
Let $G$ be a finite group acting on a finite set $X$ and let $N$ be a normal subgroup of $G$.
It is clear that $N$ naturally acts on $X$. Suppose the set of $N$-orbits are denoted by $N\verb|\|X=\{N(x)\,|\,x\in X\}$. Then the factor group $G/N$ acts on $N\verb|\|X$ and the number of orbits of $G$ on $X$
is equal to the number of orbits of $G/N$ on $N\verb|\|X$.
\end{lem}

\section{The number of orbits of ${\rm P\Gamma L}$ on $\mathcal{I}_r$}
In this section we analyze the orbits of ${\rm P\Gamma L}$ on $\mathcal{I}_r$.
As $\rm{PGL}$ is a normal subgroup of $\rm{P\Gamma L}$, by virtue of  Lemma \ref{book},
we  first count the number of orbits of $\rm{PGL}$ on the set $\mathcal{I}_r$.
According to the Cauchy-Frobenius Theorem (or named Burnside's Lemma, see \cite[Theorem 2.113]{Rotman}), we have
$$\big|{\rm PGL}\verb|\|\mathcal{I}_r\big|=\frac{1}{|{\rm PGL}|}\sum_{A\in {\rm PGL}}\big|{\rm Fix}(A)\big|
=\frac{1}{q(q^2-1)}\sum_{A\in {\rm PGL}}\big|{\rm Fix}(A)\big|,$$
where
${\rm Fix}(A)=\big\{ f \in \mathcal{I}_r\,\big|\,Af=f\big\}$
is the number of fixed points of $A\in {\rm PGL}$ in $\mathcal{I}_r$.
To find the exact value of $\big|{\rm PGL}\verb|\|\mathcal{I}_r\big|$,
it is enough to determine the number of elements of ${\rm Fix}(A)$, for each $A\in {\rm PGL}$.
To this end, in order to use some known results in the literature,
we need to consider another action of the group $\rm{PGL}$ on the set $\mathcal{I}_r$ defined by
\begin{eqnarray*}
{\rm PGL} \times \mathcal{I}_r &\longrightarrow& \mathcal{I}_r \\
(A, ~f) &\mapsto & A\circ f=\Big((bx+d)^rf\big(\frac{ax+c}{bx+d}\big)\Big)^*,
\end{eqnarray*}
where $A=\left(
 \begin{array}{cc}
 a & b \\
 c & d \\
 \end{array}
 \right)
\in \rm{PGL}$, see \cite{Reis20}.
Define
$${\rm Fix}(A,\circ)=\big\{ f \in \mathcal{I}_r\,\big|\,A\circ f=f\big\}.$$
It follows that for every $A\in {\rm PGL}$,
$${\rm Fix}(A)={\rm Fix}\big((A^T)^{-1},\circ\big),$$
where $A^T$ denotes the transpose of the matrix $A$.
Therefore, the number of orbits of $\rm{PGL}$ on the set $\mathcal{I}_r$ is equal to
\begin{eqnarray*}
\big|{\rm PGL}\verb|\|\mathcal{I}_r\big|
&=&\frac{1}{q(q^2-1)}\sum_{A\in {\rm PGL}}\big|{\rm Fix}(A)\big|\\
&=&\frac{1}{q(q^2-1)}\sum_{A\in {\rm PGL}}\big|{\rm Fix}\Big((A^T)^{-1},\circ\big)\Big|\\
&=&\frac{1}{q(q^2-1)}\sum_{A\in {\rm PGL}}\big|{\rm Fix}(A, \circ)\big|.
\end{eqnarray*}
It allows us to convert the problem of counting  $\big|{\rm Fix}(A)\big|$ to that of counting $\big|{\rm Fix}(A, \circ)\big|$.
The value of $\big|{\rm Fix}(A, \circ)\big|$ has been considered in the literature, see \cite{Reis20}.
Given $A,B,P\in {\rm PGL}$, if $PAP^{-1}=B$ then   $A$ and $B$ are called  conjugate
in ${\rm PGL}$, denoted by $A\thicksim B$. If this is the case,
according to \cite[Lemma 2.5]{Reis20} and  \cite[Theorem 2.7]{Reis20}, one has
$\big|{\rm Fix}(A)\big|=\big|{\rm Fix}(B)\big|$.
In this sense,  it is crucial to determine the conjugacy classes of the group ${\rm PGL}$.

Let $\xi$ be a primitive element of the finite
field $\mathbb{F}_{q^2}$; that is, the cyclic group $\mathbb{F}_{q^2}^*$ is generated by $\xi$, in symbols  $\mathbb{F}_{q^2}^*=\langle\xi\rangle$.
Then $\mathbb{F}_{q}^*=\langle\xi^{q+1}\rangle$, and
the set
$$\big\{1, \xi^{q-1}, \xi^{2(q-1)}, \cdots, \xi^{q(q-1)}\big\}$$
is a  transversal of $\mathbb{F}_{q}^*$ in $\mathbb{F}_{q^2}^*$.
Hence $\mathbb{F}_{q^2}^*$ is the disjoint union
$$
\mathbb{F}_{q^2}^*=\mathbb{F}_q^*\cup \xi^{q-1}\mathbb{F}_q^*\cup \xi^{2(q-1)}\mathbb{F}_q^*\cup\cdots \cup \xi^{q(q-1)}\mathbb{F}_q^*.
$$
Therefore the set of  elements of  $\mathbb{F}_{q^2}^*$ that do not belong to $\mathbb{F}_q^*$ is
$$
\mathbb{F}_{q^2}^*-\mathbb{F}_q^*=\bigcup_{i=1}^{\frac{q}{2}}\Big(\xi^{(q-1)i}\mathbb{F}_q^*\cup\xi^{-(q-1)i}\mathbb{F}_q^*\Big)
=\bigcup_{i=1}^{\frac{q}{2}}\Big(\xi^{(q-1)i}\mathbb{F}_q^*\cup\xi^{q(q-1)i}\mathbb{F}_q^*\Big).
$$
The number of conjugacy classes of ${\rm PGL}$ is presented in \cite{Group}.
The next result contains more detail information about the conjugacy classes of ${\rm PGL}$,
but which may not be readily available in the literature.
\begin{lem}\label{orbitnumber}
With  notation as given above,
there are exactly  four families of conjugacy classes of ${\rm PGL}$.
\begin{itemize}
\item[(1)] The matrix $E_2$ gives  a conjugacy class of size $1$.

\item[(2)] The matrix $$U_1=\begin{pmatrix}
1 & 1 \\
0 & 1
\end{pmatrix}$$ gives  a conjugacy class which contains $q^2-1$ elements.

\item[(3)] The matrices
$$D_{1,a}=\begin{pmatrix}
1 & 0 \\
0 & a
\end{pmatrix}\big(a\in S\big),$$
give  $\frac{q-2}{2}$ conjugacy classes,
where $S\subseteq \mathbb{F}_q^*$ satisfies $\{1\}\cup S\cup S^{-1}=\mathbb{F}_q^*$ with $S^{-1}=\{s^{-1}\,|\,s\in S\}$.
Each conjugacy class contains $q(q+1)$ elements.

\item[(4)] The matrices
$$V_{\gamma_i}=\begin{pmatrix}
0 & 1 \\
\gamma_i^{1+q} & \gamma_i+\gamma_i^q
\end{pmatrix}$$
give  $\frac{q}{2}$ conjugacy classes,
where $\gamma_i=\xi^{(q-1)i}$ for $i=1,2,\cdots,\frac{q}{2}$.
Each conjugacy class contains $q(q-1)$ elements.
\end{itemize}
\end{lem}
\begin{proof}
Its proof is somewhat long, involving  some routine and tedious computations, and  is deferred to the
Appendix.
\end{proof}

By \cite[Lemma 4.1]{Reis20} and  \cite[Theorem 4.7]{Reis20}, we immediately  have

\begin{lem}\label{orbitnumber2}
Let notation be the same as in Lemma \ref{orbitnumber}. We have
\begin{itemize}
\item[(1)] If $r$ is even, then
$$\big|{\rm Fix}(U_1, \circ)\big|=
\frac{1}{r}\sum_{d\mid\frac{r}{2} \atop \gcd(2,d)=1}\mu(d)q^{\frac{r}{2d}}.$$
If $r$ is odd, then
$$\big|{\rm Fix}(U_1, \circ)\big|=0.$$

\item[(2)] Let $a\in S$ with $D={\rm ord}(a)$, where ${\rm ord}(a)$ stands for the order of the element $a$ in the multiplicative group $\mathbb{F}_q^*$. If  $r$ is divisible by $D$, saying  $r=Dm$, then
$$
\big|{\rm Fix}(D_{1,a}, \circ)\big|=\frac{\varphi(D)}{r}\sum_{d|m \atop {\rm gcd}(d,D)=1}\mu(d)(q^{\frac{m}{d}}-1),
$$
where
$\varphi$ is the Euler's Totient function and $\mu$ is the M$\ddot{o}$bius function.

If $r$ is not divisible by $D$, then
$$\big|{\rm Fix}(D_{1,a}, \circ)\big|=0.$$

\item[(3)] Let $D={\rm ord}(V_{\gamma_i})$,
where ${\rm ord}(V_{\gamma_i})$ stands for the order of the element $V_{\gamma_i}$ in the group ${\rm PGL}$.
If $r$ is divisible by $D$, saying   $r=Dm$, then
$$
\big|{\rm Fix}(V_{\gamma_i}, \circ)\big|=\frac{\varphi(D)}{r}\sum_{d|m \atop {\rm gcd}(d,D)=1}\mu(d)\Big(q^{\frac{m}{d}}+(-1)^{\frac{m}{d}+1}\Big).
$$

If $r$ is not divisible by $D$, then
$$\big|{\rm Fix}(V_{\gamma_i}, \circ)\big|=0.$$

\end{itemize}
\end{lem}

By virtue of Lemmas \ref{orbitnumber} and \ref{orbitnumber2}, we are ready to obtain the number of orbits of ${\rm PGL}$ on
$\mathcal{I}_r$, which is the main result of this section.
\begin{Theorem}\label{orbitsize1}
Let $\varphi$ and $\mu$ denote the Euler's Totient function and the M$\ddot{o}$bius function, respectively.
Let ${\rm PGL}\verb|\|\mathcal{I}_r$
be the set of all  orbits of ${\rm PGL}$ on $\mathcal{I}_r$. Then
$$\big|{\rm PGL}\verb|\|\mathcal{I}_r\big|=\frac{1}{q(q^2-1)}\big(\mathcal{N}_0+\mathcal{N}_1+\mathcal{N}_2+\mathcal{N}_3\big),$$
where
\begin{equation*}
\mathcal{N}_0=\frac{1}{r}\sum_{d|r}\mu(d)q^{\frac{r}{d}},~~~~
\mathcal{N}_1 =
\begin{cases}
0, & 2\nmid r,\\
\vspace{0.02cm}\\
\frac{q^2-1}{r}\sum\limits_{d\mid\frac{r}{2} \atop \gcd(2,d)=1}\mu(d)q^{\frac{r}{2d}}, & 2\mid r,
\end{cases}
\end{equation*}
$$
\mathcal{N}_2 = q(q+1)\cdot\sum_{D|\gcd(r,q-1) \atop D\neq 1}\frac{\varphi^2(D)}{r}\sum_{d\mid\frac{r}{D} \atop \gcd(d, D)=1}\mu(d)\big(q^{\frac{r}{Dd}}-1\big)
$$
and
$$
\mathcal{N}_3 = \frac{q(q-1)}{2}\cdot\sum_{D\mid\gcd(r,q+1) \atop D\neq 1}\frac{\varphi^2(D)}{r}\sum_{d\mid\frac{r}{D} \atop \gcd(d, D)=1}\mu(d)\big(q^{\frac{r}{Dd}}+(-1)^{\frac{r}{Dd}+1}\big).
$$
\end{Theorem}
\begin{proof}
According to  the discussions at the beginning of this section,
the number of orbits of $\rm{PGL}$ on the set $\mathcal{I}_r$ is equal to
$$\big|{\rm PGL}\verb|\|\mathcal{I}_r\big|
=\frac{1}{q(q^2-1)}\sum_{A\in {\rm PGL}}\big|{\rm Fix}(A, \circ)\big|.$$
By  Lemma \ref{orbitnumber},  we have
\begin{eqnarray*}
\big|{\rm PGL}\verb|\|\mathcal{I}_r\big|&=&
\frac{1}{q(q^2-1)}\bigg(\big|{\rm Fix}(E_2, \circ)\big|+(q^2-1)\big|{\rm Fix}(U_1, \circ)\big|\\
&&+q(q+1)\sum_{a\in S}\big|{\rm Fix}(D_{1,a}, \circ)\big|
+q(q-1)\sum_{i=1}^{\frac{q}{2}}\big|{\rm Fix}(V_{\xi^{(q-1)i}}, \circ)\big|
\bigg).
\end{eqnarray*}
Assume that
$
\mathcal{N}_0=\big|{\rm Fix}(E_2, \circ)\big|,~
\mathcal{N}_1=(q^2-1)\cdot\big|{\rm Fix}(U_1, \circ)\big|, ~
\mathcal{N}_2=q(q+1)\cdot\sum_{a\in S}\big|{\rm Fix}(D_{1,a}, \circ)\big|
$
and
$
\mathcal{N}_3=q(q-1)\cdot\sum_{i=1}^{\frac{q}{2}}\big|{\rm Fix}(V_{\xi^{(q-1)i}}, \circ)\big|.
$
Note that an enumerative formula for the size of $\mathcal{I}_r$ (see \cite[Theorem 3.25]{lidl}) is given by
$$|\mathcal{I}_r|=\frac{1}{r}\sum_{d|r}\mu(d)q^{\frac{r}{d}}.$$
Then
\begin{eqnarray*}
\mathcal{N}_0&=&\big|{\rm Fix}(E_2, \circ)\big|
=|\mathcal{I}_r|=\frac{1}{r}\sum_{d|r}\mu(d)q^{\frac{r}{d}}.
\end{eqnarray*}
In addition, by Lemma \ref{orbitnumber2}, one has
\begin{equation*}
\mathcal{N}_1 =(q^2-1)\cdot\big|{\rm Fix}(U_1, \circ)\big|=
\begin{cases}
0, & 2\nmid r,\\
\vspace{0.02cm}\\
\frac{q^2-1}{r}\sum\limits_{d|\frac{r}{2} \atop \gcd(2,d)=1}\mu(d)q^{\frac{r}{2d}}, & 2\mid r,
\end{cases}
\end{equation*}

\begin{eqnarray*}
\mathcal{N}_2 &=&q(q+1)\cdot\sum_{a\in S}\big|{\rm Fix}(D_{1,a}, \circ)\big|\\
&=&q(q+1)\cdot\sum_{D\mid(q-1) \atop D|r, D\neq 1}\frac{\varphi^2(D)}{r}\sum_{d|\frac{r}{D} \atop \gcd(d, D)=1}\mu(d)\big(q^{\frac{r}{Dd}}-1\big)\\
&=&q(q+1)\cdot\sum_{D\mid\gcd(r,q-1) \atop D\neq 1}\frac{\varphi^2(D)}{r}\sum_{d|\frac{r}{D} \atop \gcd(d, D)=1}\mu(d)\big(q^{\frac{r}{Dd}}-1\big)
\end{eqnarray*}
and
\begin{eqnarray*}
\mathcal{N}_3 &=&q(q-1)\cdot\sum_{i=1}^{\frac{q}{2}}\big|{\rm Fix}(V_{\xi^{(q-1)i}}, \circ)\big|\\
&=&\frac{q(q-1)}{2}\cdot\sum_{i=1}^q\big|{\rm Fix}(V_{\xi^{(q-1)i}}, \circ)\big|\\
&=&\frac{q(q-1)}{2}\cdot\sum_{D\mid(q+1) \atop D\mid r, D\neq 1}\frac{\varphi^2(D)}{r}\sum_{d|\frac{r}{D} \atop \gcd(d, D)=1}\mu(d)\big(q^{\frac{r}{Dd}}+(-1)^{\frac{r}{Dd}+1}\big)\\
&=&\frac{q(q-1)}{2}\cdot\sum_{D|\gcd(r,q+1) \atop D\neq 1}\frac{\varphi^2(D)}{r}\sum_{d|\frac{r}{D} \atop \gcd(d, D)=1}\mu(d)\big(q^{\frac{r}{Dd}}+(-1)^{\frac{r}{Dd}+1}\big).
\end{eqnarray*}
We are done.
\end{proof}

\section{The number of orbits of ${\rm Gal}$ on ${\rm PGL}\backslash\mathcal{I}_r$}
In order to get the number of orbits of ${\rm P\Gamma L}$ on $\mathcal{I}_r$, by Lemma \ref{book} and
Theorem \ref{orbitsize1}, we are left  to count the number of orbits of   ${\rm Gal}$ on ${\rm PGL}\verb|\|\mathcal{I}_r$.
Recall that   the Galois group ${\rm Gal}={\rm Gal}(\mathbb{F}_{q^r}/\mathbb{F}_2)={\rm Gal}(\mathbb{F}_{2^{rn}}/\mathbb{F}_2)=\langle\sigma \rangle$ is the cyclic group of order $rn$ generated by $\sigma.$
The action of  ${\rm Gal}$   on ${\rm PGL}\verb|\|\mathcal{I}_r$ is given by
\begin{eqnarray*}
{\rm Gal}\times {\rm PGL}\verb|\|\mathcal{I}_r &\longrightarrow & {\rm PGL}\verb|\|\mathcal{I}_r\\
\big(\sigma^i, ~{\rm PGL}(f)\big)&\mapsto & \sigma^i\big({\rm PGL}(f)\big)={\rm PGL}(\sigma^if).
\end{eqnarray*}
Recall also that $n\geq 5$ is a prime number, $q=2^n$ and  $r\geq 3$ is a positive integer relatively prime to $n$.
Thus ${\rm Gal}=\langle\sigma\rangle$ has the following decomposition into direct products:
$${\rm Gal}=\langle\sigma^r\rangle\times \langle\sigma^n\rangle.$$
In order to count the number of orbits of   ${\rm Gal}$ on ${\rm PGL}\verb|\|\mathcal{I}_r$, using Lemma \ref{book} again,  we first consider the
action of $\langle\sigma^n\rangle$ on  ${\rm PGL}\verb|\|\mathcal{I}_r$.
Clearly,  the action of $\langle\sigma^n\rangle$ on ${\rm PGL}\verb|\|\mathcal{I}_r$ is given by
\begin{eqnarray*}
\langle \sigma^n\rangle\times {\rm PGL}\verb|\|\mathcal{I}_r &\longrightarrow & {\rm PGL}\verb|\|\mathcal{I}_r\\
\big(\sigma^{ni}, ~{\rm PGL}(f)\big)&\mapsto & \sigma^{ni}\big({\rm PGL}(f)\big)={\rm PGL}(\sigma^{ni}f).
\end{eqnarray*}
Observe that
$\sigma^na=a^{2^n}=a^q=a \hbox{ for any $a\in \mathbb{F}_q$,}$
which gives
$${\rm PGL}(\sigma^{ni}f)={\rm PGL}(f)~~\hbox{for any $f\in \mathcal{I}_r$}.$$
This means that $\langle\sigma^n\rangle$ fixes each ${\rm PGL}(f)$ in ${\rm PGL}\verb|\|\mathcal{I}_r$; in other words,
the set of orbits of $\langle\sigma^n\rangle$ on ${\rm PGL}\verb|\|\mathcal{I}_r$  remains ${\rm PGL}\verb|\|\mathcal{I}_r$.
By Lemma \ref{book}, the number of orbits of ${\rm Gal}$ on ${\rm PGL}\verb|\|\mathcal{I}_r$ is equal to the
number of orbits of $\langle \sigma^r\rangle$ on ${\rm PGL}\verb|\|\mathcal{I}_r$.
Since $\langle\sigma^r\rangle$ is of prime order $n$, the size of   every orbit of $\langle \sigma^r\rangle$ on ${\rm PGL}\verb|\|\mathcal{I}_r$ is equal to $1$ or $n$.
Thus it is enough to determine the number of  orbits of $\langle \sigma^r\rangle$ on ${\rm PGL}\verb|\|\mathcal{I}_r$ with size $1$.
%Thus the above fact reduces the study of the action of the group $G$ on $\Omega$ to the study of the action of the group $\langle\sigma^r\rangle$ on $\Omega$.
%So the key point is to count the number of all distinct orbits in $\Omega$ under the action of cyclic group $\langle\sigma^r\rangle$.

%For any $g\in \langle\sigma^r\rangle$, let ${\rm Fix}(g)$ denote the set of elements of ${\rm PGL}_2(F_q)f$ fixed by $g$, i.e.,
%$${\rm Fix}(g)=\big\{{\rm PGL}_2(F_q)f\in \Omega\big|g\big({\rm PGL}_2(F_q)f\big)={\rm PGL}_2(F_q)f\big\}.$$
%Since $\langle\sigma^r\rangle$ is of prime order, the set of fixed points of every element,
%except the identity element, of $\langle\sigma^r\rangle$ on $\Omega$ is the same.
%Thus it is enough to count $|{\rm Fix}(\sigma^r)|$.
\subsection{The orbits of $\langle \sigma^r\rangle$ on {\rm PGL}$\backslash\mathcal{I}_r$ with size $1$}
In this subsection we will characterize the orbits of $\langle \sigma^r\rangle$ on {\rm PGL}$\backslash\mathcal{I}_r$ with size $1$.
First note that if $\alpha$ is a root of $f(x)\in \mathcal{I}_r$, then $\sigma\alpha$ is a root of $\sigma f(x)$ and $A\alpha$ is a root of $Af(x)$,
where $A\in {\rm PGL}$ or $A\in {\rm GL}$. Please keep these facts in mind and we shall use them frequently during the following discussions.
\begin{lem}\label{equal}
Let $f\in \mathcal{I}_r$ and let $\alpha$ be a root of $f(x)$. Define a map $\tau$ as follows:
\begin{eqnarray*}
\tau: && {\rm PGL}\longrightarrow {\rm PGL}(\alpha)\\
&& ~~~~A\mapsto A\alpha,
\end{eqnarray*}
then $\tau$ is a bijection between ${\rm PGL}$ and ${\rm PGL}(\alpha)$.
In particular,  ${\rm PGL}$ and ${\rm PGL}(\alpha)$ have the same size, i.e., $\big|{\rm PGL}\big|=\big|{\rm PGL}(\alpha)\big|$.
\end{lem}
\begin{proof}
It is clear that the map $\tau$ is well-defined and surjective.
Assume that $A\alpha=B\alpha$, where $A,B\in {\rm PGL}$. Then $A^{-1}B\alpha=\alpha$.
Let $A^{-1}B=\begin{pmatrix} a & b\\c & d\end{pmatrix}$, which leads to
$\frac{a\alpha+b}{c\alpha+d}=\alpha,$
or equivalently   $c\alpha^2+(a+d)\alpha+b=0$. Since $r\geq 3$, we obtain
$c=0, a+d=0$ and $b=0,$
yielding $A^{-1}B=aE_2=E_2$ in ${\rm PGL}$ and $A=B$.
It follows that $\tau$ is injective. Therefore $\tau$ is a bijection between ${\rm PGL}$ and ${\rm PGL}(\alpha)$,
which implies that $\big|{\rm PGL}\big|=\big|{\rm PGL}(\alpha)\big|$.
\end{proof}

The next result improves \cite[Lemma 3.6]{cz} by removing the numerical condition $\gcd(r,q^3-q)=1$, which is one of the key steps
in this paper.
\begin{lem}\label{equivalent}
Let $f\in \mathcal{I}_r$ and let $\alpha$ be a root of $f(x)$. Then
${\rm PGL}(\sigma^r f)={\rm PGL}(f)$
if and only if
${\rm PGL}(\sigma^r\alpha)={\rm PGL}(\alpha)$.
\end{lem}
\begin{proof}
Suppose ${\rm PGL}(\sigma^r f)={\rm PGL}(f)$ and let
$$\mathcal{Q}=\Big\{{\rm PGL}(\alpha^{q^i})\,\big|\,0\leq i\leq r-1\Big\}.$$

\textbf{Claim 1}: There is a group action $\langle \sigma^r\rangle$ on the set $\mathcal{Q}$ defined by
\begin{eqnarray*}
\langle \sigma^r\rangle \times \mathcal{Q} &\longrightarrow & \mathcal{Q}\\
\big(\sigma^{rj}, ~{\rm PGL}(\alpha^{q^i})\big)& \mapsto & \sigma^{rj}{\rm PGL}(\alpha^{q^i})={\rm PGL}(\sigma^{rj}\alpha^{q^i}).
\end{eqnarray*}
Since ${\rm PGL}(\sigma^r f)={\rm PGL}(f)$, then there exists an element $A\in {\rm PGL}$ such that
$\sigma^r f=Af$. For any fixed $0\leq j\leq r-1$, we have
\begin{eqnarray*}
\sigma^{rj}f&=&\big(\sigma^{r(j-1)}\sigma^r\big)f=\sigma^{r(j-1)}\big(\sigma^rf\big)=\sigma^{r(j-1)}\big(Af\big)\\
&=&\sigma^{r(j-1)}(A)\sigma^{r(j-1)}(f)\\
&=&\sigma^{r(j-1)}(A)\sigma^{r(j-2)}(A)\sigma^{r(j-2)}(f)\\
&=&\cdots\\
&=&\sigma^{r(j-1)}(A)\cdots\sigma^{r}(A)\sigma^{r}(f)\\
&=&\sigma^{r(j-1)}(A)\cdots\sigma^{r}(A)\sigma^{r}(A)Af\\
&\in&{\rm PGL}(f).
\end{eqnarray*}
Hence
$${\rm PGL}(\sigma^{rj}f)={\rm PGL}(f),$$
yielding that there exists  $B\in {\rm PGL}$ such that $B(\sigma^{rj}f)=f$.
Consequently,  $B(\sigma^{rj}\alpha^{q^i})$ is a root of $f$, giving
$B(\sigma^{rj}\alpha^{q^i})=\alpha^{q^s},$
where $0\leq s\leq r-1$. Therefore ${\rm PGL}(\sigma^{rj}\alpha^{q^i})\in \mathcal{Q}$.

On the other hand, it is easy to see that for $0\leq i\leq r-1$, we have
$\sigma^{0}{\rm PGL}(\alpha^{q^i})={\rm PGL}(\sigma^{0}\alpha^{q^i})={\rm PGL}(\alpha^{q^i})$
and
\begin{eqnarray*}
&&(\sigma^{rj_1}\sigma^{rj_2}){\rm PGL}(\alpha^{q^i})\\
&=&{\rm PGL}\big((\sigma^{rj_1}\sigma^{rj_2})\alpha^{q^i}\big)={\rm PGL}\big(\sigma^{rj_1}(\sigma^{rj_2}\alpha^{q^i})\big)\\
&=&\sigma^{rj_1}{\rm PGL}(\sigma^{rj_2}\alpha^{q^i})=\sigma^{rj_1}\big(\sigma^{rj_2}{\rm PGL}(\alpha^{q^i})\big),
\end{eqnarray*}
where $j_1, j_2$ are two positive integers satisfying $0\leq j_1, j_2 \leq n-1$.
Claim $1$ is thus proved.

\textbf{Claim 2}: $\big| \mathcal{Q}\big|$ is a divisor of $r$, i.e., $\big| \mathcal{Q}\big|\mid r$.
To this end,
let $$\overline{\mathcal{Q}}=\Big\{A(\alpha^{q^i})\,\big|\,A\in {\rm PGL}, ~0\leq i\leq r-1\Big\}.$$
Then $\overline{\mathcal{Q}}$ is the set of all roots of the polynomials in ${\rm PGL}(f)$,  so
$$\big|\overline{\mathcal{Q}}\big|=r\cdot \big|{\rm PGL}(f)\big|=r\cdot\big[{\rm PGL}: {\rm Stab}_{\rm PGL}(f)\big].$$
In addition,   the set $\mathcal{Q}$ can be rewritten as
$$\mathcal{Q}=\big\{{\rm PGL}(\alpha^{q^{j_1}}),{\rm PGL}(\alpha^{q^{j_2}}),\cdots,{\rm PGL}(\alpha^{q^{j_{|\mathcal{Q}|}}})\big\}.$$
Using Lemma \ref{equal},  we have
\begin{eqnarray*}
\big|\overline{\mathcal{Q}}\big|&=& \big|{\rm PGL}(\alpha^{q^{j_1}})\big|+\big|{\rm PGL}(\alpha^{q^{j_2}})\big|+\cdots+
\big|{\rm PGL}(\alpha^{q^{j_{|\mathcal{Q}|}}})\big|\\
&=&\underbrace{\big|{\rm PGL}\big|+\big|{\rm PGL}\big|+\cdots+
\big|{\rm PGL}\big|}_{|\mathcal{Q}|~ \mbox{times}}\\
&=&\big|{\rm PGL}\big|\cdot \big|\mathcal{Q}\big|.
\end{eqnarray*}
Thus,
$$\big|\overline{\mathcal{Q}}\big|=r\cdot\big[{\rm PGL}: {\rm Stab}_{\rm PGL}(f)\big]=\big|{\rm PGL}\big|\cdot \big|\mathcal{Q}\big|.$$
That is to say,
$$r=\big|\mathcal{Q}\big|\cdot\frac{\big|{\rm PGL}\big|}{\big[{\rm PGL}: {\rm Stab}_{\rm PGL}(f)\big]}
=\big|\overline{\mathcal{Q}}\big| \big[{\rm Stab}_{\rm PGL}(f)\big],$$
which shows
$\big| \mathcal{Q}\big|\mid r.$
The proof of Claim $2$ is completed.

According to the two claims above, there is a group action $\langle \sigma^r\rangle$
on the set $\mathcal{Q}$ and the size of $\mathcal{Q}$ is a divisor of $r$.
Since $\langle\sigma^r\rangle$ is of prime order $n$, the size of every orbit of $\langle\sigma^r\rangle$ on $\mathcal{Q}$ is equal to
$1$ or $n$. From $\gcd(r,n)=1$ and $\big| \mathcal{Q}\big|\mid r$ we obtain
$\gcd\big(|\mathcal{Q}|, n\big)=1,$
i.e.,
$n \nmid |\mathcal{Q}|.$
It follows that there exists an orbit of $\langle\sigma^r\rangle$ on $\mathcal{Q}$
with size $1$. Suppose that this orbit with size $1$ is ${\rm PGL}(\alpha^{q^{j_t}})$ with $0\leq t \leq r-1$.
Then
$$\sigma^r{\rm PGL}(\alpha^{q^{j_t}})={\rm PGL}(\sigma^r\alpha^{q^{j_t}})={\rm PGL}(\alpha^{q^{j_t}}).$$
Hence, there exists  $D\in {\rm PGL}$ satisfying
$D(\sigma^r\alpha^{q^{j_t}})=\alpha^{q^{j_t}},$
implying
$(D\sigma^r\alpha)^{q^{j_t}}=\alpha^{q^{j_t}}.$
Therefore
$D\sigma^r\alpha=\alpha,$
which implies that
${\rm PGL}(\sigma^r\alpha)={\rm PGL}(\alpha).$

Conversely, suppose that ${\rm PGL}(\sigma^r\alpha)={\rm PGL}(\alpha)$.
Then there is a matrix $A\in {\rm PGL}$ such that $A(\sigma^r\alpha)=\alpha$.
Note that $A(\sigma^r\alpha)$ is a root of $A(\sigma^rf)$, and then we obtain $A(\sigma^r f)=f$,
which implies that ${\rm PGL}(\sigma^r f)={\rm PGL}(f).$
We are done.
\end{proof}
%We have shown that  is fixed by $\langle \sigma^r\rangle$ if and only if ${\rm PGL}(\alpha)$ is fixed by $\langle\sigma^r\rangle$,
%where $\alpha$ is a root of $f(x)$.
To count the number of
${\rm PGL}(f)\in {\rm PGL}\verb|\|\mathcal{I}_r$ that are fixed by $\langle \sigma^r\rangle$,  we need to use the
affine general linear group ${\rm AGL}$. The affine general linear group ${\rm AGL}$ can be viewed naturally as a subgroup of ${\rm PGL}$.
Hence, the group ${\rm AGL}$ acts on the set $\mathcal{S}$ naturally.
Let
$$
{\rm AGL}\verb|\|\mathcal{S}=\big\{{\rm AGL}(\alpha)\,\big|\,\alpha\in \mathcal{S}\big\}
$$
be the set of all orbits of ${\rm AGL}$ on $\mathcal{S}$. Then the cyclic group $\langle \sigma^r\rangle$ acts on
${\rm AGL}\verb|\|\mathcal{S}$ in the following way:
\begin{equation}\label{AGL-action}
\langle \sigma^r\rangle\times {\rm AGL}\verb|\|\mathcal{S}  \longrightarrow  {\rm AGL}\verb|\|\mathcal{S},~~~~
\\
\big(\sigma^{ri}, ~{\rm AGL}(\alpha)\big)\mapsto  \sigma^{ri}\big({\rm AGL}(\alpha)\big)={\rm AGL}(\sigma^{ri}\alpha).
\end{equation}
It is not hard to verify that this is indeed a group action.
We now turn to consider the orbit ${\rm PGL}(\alpha)$ where $\alpha\in \mathcal{S}$.
There is an action of ${\rm AGL}$ on ${\rm PGL}(\alpha)$:
\begin{eqnarray*}
{\rm AGL}\times {\rm PGL}(\alpha) &\longrightarrow & {\rm PGL}(\alpha)\\
(C,  ~A\alpha)&\mapsto & CA\alpha.
\end{eqnarray*}
Therefore, ${\rm PGL}(\alpha)$ is the disjoint union of ${\rm AGL}$-orbits. Indeed, one can easily check that there are exactly
$q+1$ right cosets of ${\rm AGL}$ in ${\rm PGL}$ and
$$
t_0=E_2=\left(
                         \begin{array}{cc}
                           1 &0 \\
                           0 & 1 \\
                         \end{array}
                       \right),~~ t_1=\left(
                         \begin{array}{cc}
                           0 & 1 \\
                           1 & 0 \\
                         \end{array}
                       \right)
\hbox{~~and~~}
t_\gamma=\left(
                         \begin{array}{cc}
                           0 & 1 \\
                           1 & \gamma \\
                         \end{array}
                      \right)    \hbox{~~for any $\gamma\in \mathbb{F}_q^*$}
$$
consists of a right coset representative of ${\rm AGL}$ in ${\rm PGL}$.
The coset decomposition
$${\rm PGL}={\rm AGL}t_0\bigcup{\rm AGL}t_1\bigcup_{\gamma\in \mathbb{F}_q^*}{\rm AGL}t_\gamma$$
gives rise to the orbit decomposition of ${\rm PGL}(\alpha)$ into ${\rm AGL}$-orbits
$${\rm PGL}(\alpha)={\rm AGL}(t_0\alpha)\bigcup{\rm AGL}(t_1\alpha)\bigcup_{\gamma\in \mathbb{F}_q^*}{\rm AGL}(t_\gamma\alpha).$$
We have arrived at the following result (which has been appeared previously in \cite{cz}).
\begin{lem}\label{partition}
Let $\alpha\in \mathcal{S}$. Then
$${\rm PGL}(\alpha)=\bigcup_{\gamma\in \mathbb{F}_q}{\rm AGL}\Big(\frac{1}{\alpha+\gamma}\Big)\bigcup{\rm AGL}(\alpha),$$
is a partition of ${\rm PGL}(\alpha)$ into ${\rm AGL}$-orbits.
\end{lem}
Lemma \ref{partition} implies that
\begin{equation*}
\begin{split}
{\rm PGL}\big(\sigma^r(\alpha)\big)&=\bigcup_{\gamma\in
\mathbb{F}_q}{\rm AGL}\Big(\frac{1}{\sigma^r(\alpha)+\gamma}\Big)\bigcup{\rm AGL}\big(\sigma^r(\alpha)\big)\\
&=\bigcup_{\gamma\in
\mathbb{F}_q}{\rm AGL}\Big(\frac{1}{\sigma^r(\alpha)+\sigma^r(\gamma)}\Big)\bigcup{\rm AGL}\big(\sigma^r(\alpha)\big)\\
&=\bigcup_{\gamma\in
\mathbb{F}_q}{\rm AGL}\Big(\sigma^r\Big(\frac{1}{\alpha+\gamma}\Big)\Big)\bigcup{\rm AGL}\big(\sigma^r(\alpha)\big).\\
\end{split}
\end{equation*}
Suppose now that ${\rm PGL}(\alpha)$ is fixed by the cyclic group $\langle\sigma^r\rangle$,
i.e., ${\rm PGL}(\sigma^r\alpha)={\rm PGL}(\alpha)$.
In this case, the cyclic group $\langle \sigma^r\rangle$ acts on the set of ${\rm AGL}$-orbits
$${\rm AGL}\verb|\|{\rm PGL}(\alpha)=\Big\{{\rm AGL}(\alpha), {\rm AGL}\Big(\frac{1}{\alpha+\gamma}\Big)\,\Big|\,\gamma\in \mathbb{F}_q\Big\}$$
in the way given in (\ref{AGL-action}).

\begin{lem}\label{fixpoint}
{\rm(\cite[Lemma 3.8.]{cz})} Let $n\geq5$ be a prime number. If ${\rm PGL}(\sigma^r\alpha)={\rm PGL}(\alpha)$,
then there exists a fixed point of $\langle\sigma^r\rangle$  on ${\rm AGL}\verb|\|{\rm PGL}(\alpha)$. In other words,
either ${\rm AGL}(\sigma^r\alpha)={\rm AGL}(\alpha)$ or
${\rm AGL}\big(\sigma^r(\frac{1}{\alpha+\gamma})\big)={\rm AGL}\big(\frac{1}{\alpha+\gamma}\big)$
for some $\gamma\in \mathbb{F}_q$.
\end{lem}
%\begin{proof}
%Note that the following properties hold:
%(1) $|{\rm AGL}\verb|\|{\rm PGL}(\alpha)|=q+1$; (2) $n$ does not divide $q+1$;
%(3) the size of each orbit of $\langle \sigma^r\rangle$ on ${\rm AGL}\verb|\|{\rm PGL}(\alpha)$ is either $1$ or $n$.
%We get the required result.
%\end{proof}

By Lemma \ref{equivalent}, we derive the next result which   is crucial to our enumeration.
\begin{lem}\label{nexttolastlem}
Let $f\in \mathcal{I}_r$.
Then
${\rm PGL}(\sigma^rf)={\rm PGL}(f)$
if and only if there is a polynomial $g(x)\in {\rm PGL}(f)$ such that
$g(x)$ divides $x^{2^r}+x.$
\end{lem}
\begin{proof}
The proof is essentially the same as that given in \cite[Lemma 3.9]{cz}, since we have established
Lemma \ref{equivalent}.
\end{proof}

\subsection{The number of orbits of $\langle \sigma^r\rangle$ on {\rm PGL}$\backslash\mathcal{I}_r$ with size $1$}
Now we are in a position  to determine the number of orbits of $\langle \sigma^r\rangle$ on ${\rm PGL}\verb|\|\mathcal{I}_r$ with size $1$.
For convenience, we adopt the following notation throughout this subsection.
\begin{itemize}
\item[] $A_1=\begin{pmatrix} 1 & 0 \\0 & 1\end{pmatrix}=E_2,~
A_2=\begin{pmatrix} 0 & 1 \\1 & 0\end{pmatrix},~
A_3=\begin{pmatrix} 1 & 1 \\0 & 1\end{pmatrix},~A_4=\begin{pmatrix} 1 & 0 \\1 & 1\end{pmatrix},~
A_5=\begin{pmatrix} 1 & 1 \\1 & 0\end{pmatrix},~
A_6=\begin{pmatrix} 0 & 1 \\1 & 1\end{pmatrix}.$

\item[] $G=\Big\{A_1, A_2, A_3, A_4, A_5, A_6\Big\}.$

\item[] $\mathcal{X}=\Big\{f(x)\in \mathcal{I}_r\,\Big|\,\hbox{$f(x)$ divides $x^{2^r}+x$}\Big\}.$
\item[] $\Delta_i=\Big\{h(x)\in \mathcal{X}\,\big|\,A_ih(x)=h(x)\Big\} ~\hbox{for}~ 2\leq i\leq 6.$
\item[] $\Delta_7=\Big\{h(x)\in \mathcal{X}\,\big|\,A_ih(x)\neq h(x), ~\hbox{for any}~ 2\leq i\leq 6\Big\}.$
\item[] $\mathcal{G}_f=\Big\{A_1f=f, A_2f, A_3f, A_4f, A_5f, A_6f\Big\}, ~f\in \mathcal{X}.$
\end{itemize}
%\big{|}\mathcal{X}\big|&=&
%\sum_{e\in {\rm E}(r,q)}\frac{\varphi(e)}{r}%=\frac{1}{r}\sum_{e\in {\rm E}(r,q)}\varphi(e)
%=\frac{1}{r}\sum_{d|r}\big(2^{\frac{r}{d}}-1\big)\mu(d),
%where $\varphi$ is the Euler's function and $\mu$ is the M\"{o}bius function.

The following  result reveals that if ${\rm PGL}(f)$ contains a polynomial that divides $x^{2^r}+x$, then
${\rm PGL}(f)$ contains  $2,3$ or $6$ such polynomials.

\begin{lem}\label{lastlem}
Suppose that $f(x)\in \mathcal{I}_r$ such that $f(x)$ divides $x^{2^r}+x$. Then
$$\Big\{h(x)\,\Big|\,h(x)\in {\rm PGL}(f), ~\hbox{$h(x)$ divides $x^{2^r}+x$}\Big\}=\mathcal{G}_f;$$
in particular,
$$\Big|\Big\{h(x)\,\Big|\,h(x)\in {\rm PGL}(f), ~\hbox{$h(x)$ divides $x^{2^r}+x$}\Big\}\Big|=\big|\mathcal{G}_f\big|.$$
\end{lem}
\begin{proof}
For simplifying notation, let  $\Delta=\big\{h(x)\,\big|\,h(x)\in {\rm PGL}(f), ~\hbox{$h(x)$ divides $x^{2^r}+x$}\big\}$.
Let $\alpha$ be   a root of $f(x)$, which gives $\alpha^{2^r}=\alpha$ since $f(x)$ divides $x^{2^r}+x$.
Observe that
\begin{equation*}
\begin{split}
\Delta=&\Big\{h(x)\,\Big|\,h(x)\in {\rm PGL}(f), ~\hbox{$h(x)$ divides $x^{2^r}+x$}\Big\}\\
=&\Big\{Af(x)\,\Big|\,   \hbox{$Af(x)$ divides $x^{2^r}+x$}\Big\}\\
%=&\Big|\Big\{A\in {\rm PGL}\,\Big|\,   \hbox{$Af(x)$ divides $x^{2^r}+x$}\Big\}\Big|\\
=&\Big\{Af(x)\Big|\,   (A\alpha)^{2^r}+A\alpha=0\Big\}.
\end{split}
\end{equation*}
%\begin{eqnarray*}
%{\rm PGL}(h)={\rm PGL}(f)
%&\Leftrightarrow & h(x)=Ag(x), \exists~ A\in {\rm PGL}_2(F_q)\\
%&\Leftrightarrow & (A\alpha)^{2^r}+A\alpha=0,\\
%\end{eqnarray*}
%and we obtain that
%$$\big|\Delta\big|=\Big|\big\{A\in {\rm PGL}_2(F_q)\big|(A\alpha)^{2^r}+A\alpha=0\big\}\Big|.$$
Assume that $A=\begin{pmatrix} a & b \\c & d\end{pmatrix}.$
Then
\begin{eqnarray*}
(A\alpha)^{2^r}+A\alpha=0
&\Leftrightarrow & \Big(\frac{a\alpha+b}{c\alpha+d}\Big)^{2^r}+\Big(\frac{a\alpha+b}{c\alpha+d}\Big)=0\\
&\Leftrightarrow & \frac{a^{2^r}\alpha^{2^r}+b^{2^r}}{c^{2^r}\alpha^{2^r}+d^{2^r}}+\frac{a\alpha+b}{c\alpha+d}=0\\
&\Leftrightarrow & \frac{a^{2^r}\alpha+b^{2^r}}{c^{2^r}\alpha+d^{2^r}}+\frac{a\alpha+b}{c\alpha+d}=0\\
&\Leftrightarrow & (ca^{2^r}+ac^{2^r})\alpha^2+(da^{2^r}+bc^{2^r}+ad^{2^r}+cb^{2^r})\alpha+(bd^{2^r}+db^{2^r})=0\\
&\Leftrightarrow &
\begin{cases}
ca^{2^r}+ac^{2^r}=0,\\
da^{2^r}+bc^{2^r}+ad^{2^r}+cb^{2^r}=0,\\
bd^{2^r}+db^{2^r}=0.
\end{cases}
\end{eqnarray*}
We have to consider three cases separately.

Case 1: $a\neq 0, ~c=0$.
Since $A$ is invertible, one must have $d\neq0$. From the second equality we have $a=d$.
If $b\neq 0$, then $b=d$. Hence, there are   two cases:
$$b=c=0,~a=d\neq0~~\hbox{and}~~c=0, ~a=b=d\neq0.$$
Therefore in this case
$$A=\begin{pmatrix} 1 & 0 \\0 & 1\end{pmatrix}~\mbox{or}~
\begin{pmatrix} 1 & 1 \\0 & 1\end{pmatrix}.$$
Case 2: $c\neq 0, ~a=0$.
Since $A$ is invertible, $b\neq0$ and $ c\neq 0$. By the second equality we obtain   $b=c$.
If $d\neq 0$, then $b=d$. Hence, there are   two cases:
$$b=c\neq 0,~a=d=0~~\hbox{and}~~a=0, ~b=c=d\neq0.$$
Therefore
$$A=\begin{pmatrix} 0 & 1 \\1 & 0\end{pmatrix}~\mbox{or}~
\begin{pmatrix} 0 & 1 \\1 & 1\end{pmatrix}.$$
Case 3: $c\neq 0, ~a\neq0$.
From the first equality we get $a=c$. We consider three subcases separately.

Subcase 3.1: $b=0, ~d\neq 0$. From the second equality,  we have $a=d$.

Subcase 3.2: $b\neq0, ~d=0$. From the second equality, we have $b=c$.

Subcase 3.3: $b\neq0, ~d\neq0$. From the last equality,  we have $b=d$. However, the determinate of $A$ is  $ad-bc=0$. This is impossible.

Hence, there are  two cases:
$$b=0,~a=c=d\neq 0~~\hbox{and}~~d=0,~a=b=c\neq0.$$
Therefore in this case
$A=\begin{pmatrix} 1 & 0 \\1 & 1\end{pmatrix}~\mbox{or}~
\begin{pmatrix} 1 & 1 \\1 & 0\end{pmatrix}.$

In conclusion, we have
\begin{eqnarray*}
\Delta&=&\Bigg\{\begin{pmatrix} 1 & 0 \\0 & 1\end{pmatrix}f(x),
\begin{pmatrix} 1 & 1 \\0 & 1\end{pmatrix}f(x),
\begin{pmatrix} 0 & 1 \\1 & 0\end{pmatrix}f(x),
\begin{pmatrix} 0 & 1 \\1 & 1\end{pmatrix}f(x),
\begin{pmatrix} 1 & 0 \\1 & 1\end{pmatrix}f(x),
\begin{pmatrix} 1 & 1 \\1 & 0\end{pmatrix}f(x)
\Bigg\}\\
&=&\Big\{A_1f=f, A_2f, A_3f, A_4f, A_5f, A_6f\Big\}\\
&=&\mathcal{G}_f.
\end{eqnarray*}
Hence we get at once that $\big|\Delta\big|=\big|\mathcal{G}_f\big|$, which is the required result.
\end{proof}

%\begin{Corollary}
%With notations as above. Let $f\in {\rm I}_r$, then
%$$\sigma^r\big({\rm PGL}(f)\big)={\rm PGL}(f)$$
%if and only if
%there is a polynomial $g(x)\in {\rm AGL}(f)$ such that
%$$g(x)\big|(x^{2^r}+x).$$
%\end{Corollary}
We now provide some properties of the sets $G, \mathcal{X}$ and $\Delta_i~(i=2,3,\cdots,7)$.
We first observe  that $G$ is a group of order $6$ and this group is isomorphism to the symmetric group $S_3$ of degree $3$.
Since $r\geq3$,
$f(x)\in \mathcal{I}_r$ divides $x^{2^r}+x$  if and only if   $f(x)\in \mathcal{I}_r$ divides $x^{2^r-1}-1$ .
%Let ${\rm ord}(f)$ denote the order of the polynomial $f$ (see \cite[Definition 3.2]{lidl}).
%It follows from \cite[Lemma 3.6]{lidl} that $f(x)$ divides $x^{2^r}+x$ if and only if ${\rm ord}(f)$ divides $2^r-1$.
%The set ${\rm E}(r,q)$ is defined by
%\begin{equation}\label{erq}
%$${\rm E}(r,q)=\Big\{e \,\Big|\,\hbox{$e>1$ is an integer dividing $2^r-1$~but~$e$ does not divide $q^{d}-1$~for any~$1\leq d<r$}\Big\}.$$
%\end{equation}
%Then according to \cite[Theorem 3.5]{lidl},  the number of polynomials  $f(x)\in \mathcal{I}_r$
%such that $f(x)$ divides $x^{2^r}+x$ is equal to
%$$
%\big{|}\mathcal{X}\big|=\Big{|}\Big\{f(x)\in \mathcal{I}_r\,\Big|\,\hbox{$f(x)$ divides $x^{2^r}+x$}\Big\}\Big|=
%\sum_{e\in {\rm E}(r,q)}\frac{\varphi(e)}{r},$$
%where $\varphi$ is the Euler's Totient function.

The size of $\mathcal{X}$ has been determined explicitly in \cite[Lemma 3.10]{cz} in terms of the  M$\ddot{o}$bius function.

\begin{lem}\label{firstnumber}
With   notation as given above, we have
$$\big{|}\mathcal{X}\big|=\Big{|}\Big\{f(x)\in \mathcal{I}_r\,\Big|\,\hbox{$f(x)$ divides $x^{2^r}+x$}\Big\}\Big|=
\frac{1}{r}\sum_{d|r}\big(2^{\frac{r}{d}}-1\big)\mu(d),$$
where $\mu$ is the M$\ddot{o}$bius function.
\end{lem}
%\begin{proof}
%Let $d$ be a positive integer.
%Suppose that
%$$\Omega_d(x)=\prod_{f(x)\in \mathcal{S}_d(x)}f(x),$$
%where
%$$\mathcal{S}_d(x)=\Big\{f(x)\in \mathcal{I}_d\,\Big|\,\hbox{$f(x)$ divides $x^{2^r-1}-1$}\Big\}.$$
%Note that $\gcd(2^r-1, q)=\gcd(2^r-1, 2^n)=1$ and $\gcd(r,n)=1$ and we get that
%${\rm ord}_{2^r-1}(q)=r$, which shows that $d|r$. Then
%$$x^{2^r-1}-1=\prod_{d|r}\Omega_d(x).$$
%So
%$$2^r-1=\sum_{d|r}d\big|\mathcal{S}_d(x)\big|.$$
%In virtue of the M$\ddot{o}$bius inversion formula we obtain that
%$$r\big|\mathcal{S}_r(x)\big|=\frac{1}{r}\sum_{d|r}\big(2^{\frac{r}{d}}-1\big)\mu(d).$$
%Hence we have
%$$\Big{|}\Big\{f(x)\in \mathcal{I}_r\,\Big|\,\hbox{$f(x)$ divides $x^{2^r}+x$}\Big\}\Big|=\big|\mathcal{S}_r(x)\big|=
%\frac{1}{r}\sum_{d|r}\big(2^{\frac{r}{d}}-1\big)\mu(d),$$
%where $\mu$ is the M$\ddot{o}$bius function.
%\end{proof}
Let $M$ be a subgroup of ${\rm PGL}$ and $f(x)\in \mathcal{I}_r$.
If $A\circ f(x)=f(x)$ for any $A\in M$, then according to \cite[Theorem 1.3]{Reis20} we see that $M$ must be a cyclic subgroup of ${\rm PGL}$.
Naturally, we have an analogous result about the group action involved.

\begin{lem}\label{cyclic}
Let $M$ be a subgroup of ${\rm PGL}$ and $f(x)\in \mathcal{I}_r$.
If $Af(x)=f(x)$ for any $A\in M$, then $M$ is a cyclic subgroup of ${\rm PGL}$.
\end{lem}
\begin{proof}
Assume that $M_0=\big\{(A^T)^{-1}\,\big|\,A\in M\big\}$. Then $M_0$ is a subgroup of ${\rm PGL}$, which is isomorphic to $M$.
For any $B\in M_0$, there exists a matrix $A\in M$ such that $B=(A^T)^{-1}$. Thus
$$B\circ f(x)=Af(x)=f(x).$$
It follows from \cite[Theorem 1.3]{Reis20}  that $M_0$ is cyclic, which implies that $M$ is also cyclic.
\end{proof}

\begin{lem}\label{beforelastlem}
With   notation as given above, we have
\begin{itemize}
\item[(1)] $\Delta_5=\Delta_6$.

\item[(2)] $|\Delta_2|=|\Delta_3|=|\Delta_4|$.

\item[(3)] $\Delta_i\cap \Delta_j=\emptyset$ for any $i,j\in \{2,3,4,5,7\}$ with $i\neq j$.

\item[(4)] $\mathcal{X}=\Delta_2\cup\Delta_3\cup\Delta_4\cup\Delta_5\cup\Delta_7$. In particular,
$\big{|}\mathcal{X}\big|=3|\Delta_2|+|\Delta_5|+|\Delta_7|.$
\end{itemize}
\end{lem}
\begin{proof}
(1) Since $A_5^2=A_6$ and $A_6^2=A_5$, it is easy to see that if $A_5f=f$ then $A_6f=A_5^2f=f$;
if $A_6f=f$  then $A_5f=A_6^2f=f$. This shows that $\Delta_5=\Delta_6$.

(2) Note that $G$ is a group of order $6$ and this group is isomorphism to the symmetric group of degree $3$.
Then $G$ has three conjugacy classes given as follows:
$$\big\{A_1\big\}, ~\big\{A_2, A_3, A_4\big\}, ~\big\{A_5, A_6\big\}.$$
Additionally, there is a group action $G$ on the set $\mathcal{X}$:
\begin{eqnarray*}
G\times \mathcal{X} &\longrightarrow & \mathcal{X}\\
(A_i, ~f)&\mapsto & A_if.
\end{eqnarray*}
Thus for $i=1,2,\cdots,6$
$${\rm Fix}(A_i)=\big\{f\in \mathcal{X}\,\big|\,A_if=f\big\}=\Delta_i.$$
Assume that two matrices $A_i$ and $A_j$ of $G$ belong to the same conjugacy class, i.e., there exists a matrix $A_k\in G$ such that
$A_k^{-1}A_iA_k=A_j,$
where $1\leq i,j,k\leq 6$. Then
$${\rm Fix}(A_j)=A_k{\rm Fix}(A_i)\Rightarrow \Delta_j=A_k\Delta_i,$$
which implies that $\big|\Delta_i\big|=\big|\Delta_j\big|.$ Therefore we have
$|\Delta_2|=|\Delta_3|=|\Delta_4|.$

(3) Suppose that $f\in \Delta_i\cap \Delta_j$ with $i,j\in \{2,3,4,5,7\}$ and $i\neq j$.
Then $A_if=A_jf=f$. Let $\langle A_i, A_j\rangle$ be the subgroup of $G$ generated by $A_i, A_j$
(which is the smallest subgroup of $G$ containing $A_i$ and $A_j$). Thus the following holds:
$$Af=f,~\mbox{for any} ~A\in \langle A_i, A_j\rangle.$$
Since $r\geq 3$, according to Lemma \ref{cyclic},  the subgroup $\langle A_i, A_j\rangle$ is a cyclic subgroup of $G$.
Observe that for any $i,j\in \{2,3,4,5,7\}$ with $i\neq j$,    $\langle A_i, A_j\rangle$ cannot be a cyclic subgroup of $G$.
We have obtained a contradiction. Hence $\Delta_i\cap \Delta_j=\emptyset$ for any $i,j\in \{2,3,4,5,7\}$ with $i\neq j$.

(4) It follows from  (3) that  there is a decomposition of $\mathcal{X}$,
$$\mathcal{X}=\Delta_2\cup\Delta_3\cup\Delta_4\cup\Delta_5\cup\Delta_7,$$
and by (1) and (2) we get
$$\big{|}\mathcal{X}\big|=3|\Delta_2|+|\Delta_5|+|\Delta_7|.$$
\end{proof}

We have seen that  the number of monic irreducible polynomials of degree
$r$ over $\mathbb{F}_q$  that divide $x^{2^r}+x$ is equal to
$\frac{1}{r}\sum_{d|r}\big(2^{\frac{r}{d}}-1\big)\mu(d).$
Our next goals are to determine the number $|\mathcal{G}_f|$ and  obtain the number of
orbits with size $1$.

\begin{lem}\label{afterlastlem}
With  notation as given above, we then have

%(1) The group $G$ acts transitively on the set $\mathcal{G}_f$ and
%$$|\mathcal{G}_f|=\big[G: {\rm Stab}_{\rm G}(f)\big],$$
%where ${\rm Stab}_{\rm G}(f)=\big\{A\in G\big|Af=f\big\}.$
\begin{itemize}
\item[(1)] If $f\in \Delta_2$, then
%${\rm Stab}_{\rm G}(f)=\langle A_2\rangle$ and
%$|\mathcal{G}_f|=3$ and
$$\mathcal{G}_f=\big\{f, A_3f, A_4f\big\},$$
where $A_3f\in \Delta_4, ~A_4f\in \Delta_3$. In particular, $\big|\mathcal{G}_f\big|=3$.

\item[(2)] If $f\in \Delta_3$, then
$$\mathcal{G}_f=\big\{f, A_2f, A_4f\big\},$$
where $A_2f\in \Delta_4, ~A_4f\in \Delta_2$. In particular, $\big|\mathcal{G}_f\big|=3$.

\item[(3)] If $f\in \Delta_4$, then
$$\mathcal{G}_f=\big\{f, A_2f, A_3f\big\},$$
where $A_2f\in \Delta_3, ~A_3f\in \Delta_2$. In particular, $\big|\mathcal{G}_f\big|=3$.

\item[(4)] If $f\in \Delta_5$, then
$$\mathcal{G}_f=\big\{f, A_2f\big\},$$
where $A_2f\in \Delta_5$. In particular, $\big|\mathcal{G}_f\big|=2$.

\item[(5)] If $f\in \Delta_7$, then
$$\mathcal{G}_f=\big\{A_1f=f, A_2f, A_3f, A_4f, A_5f, A_6f\big\}.$$
In particular, $\big|\mathcal{G}_f\big|=6$.

\item[(6)] If $g\notin {\rm PGL}(f)$, then
$$\mathcal{G}_f\cap\mathcal{G}_g=\emptyset.$$
\end{itemize}
\end{lem}
\begin{proof}
(1).  By straightforward calculations we have
$$(A_3^{-1}A_5)f=(A_3A_5)f=A_2f=f~~\hbox{and}~~(A_4^{-1}A_6)f=(A_4A_6)f=A_2f=f,$$
which implies
$A_3f=A_5f~\hbox{and}~A_4f=A_6f.$
On the other hand, by
$$A_4(A_3f)=(A_4A_3)f=A_5f=A_3f~~\hbox{and}~~A_3(A_4f)=(A_3A_4)f=A_6f=A_4f,$$
we obtain
$$A_3f=A_5f\in \Delta_4~~\hbox{and}~~A_4f=A_6f\in \Delta_3.$$
Since the intersection of any two of $\Delta_2, \Delta_3, \Delta_4$ is empty, we get
$A_3f \neq f, ~A_3f\neq A_4f~\hbox{and}~A_4f\neq f.$
Thus
$$\mathcal{G}_f=\big\{f, A_3f, A_4f\big\}.$$

(2) and (3).  By the method analogous to that used in the proof of (1), we  obtain the desired results.

(4).  Clearly, $A_6f=A_5^2f=A_5f=f$. In addition,
$$(A_2^{-1}A_4)f=(A_2A_4)f=A_5f=f~~\hbox{and}~~(A_4^{-1}A_3)f=(A_4A_3)f=A_5f=f,$$
which yields
$A_2f=A_3f=A_4f.$
Then from $A_5(A_2f)=(A_5A_2)f=A_3f=A_2f$ we obtain   $A_2f\in \Delta_5$.
Since $A_2f\neq f$, it is easy to see that
$$\mathcal{G}_f=\big\{f, A_2f\big\}.$$

(5) and (6).  They are  obvious by the definition of $\Delta_7$ and $\mathcal{G}_f$.
\end{proof}

The number of orbits of ${\rm P\Gamma L}$ on $\mathcal{I}_r$ with size $1$ can be represented in terms of the values of
$|\mathcal{X}|$, $|\Delta_2|$ and $|\Delta_5|$, as we show below.

\begin{Theorem}\label{szero}
Let $s_0$ be the number of orbits of ${\rm P\Gamma L}$ on $\mathcal{I}_r$ with size $1$. Then
$$s_0=\frac{1}{6}\big(\big{|}\mathcal{X}\big|+3|\Delta_2|+2|\Delta_5|\big).$$
\end{Theorem}
\begin{proof}
Combining Lemmas \ref{lastlem}, \ref{beforelastlem}   and   \ref{afterlastlem}, we see that
\begin{eqnarray*}
s_0&=&\frac{1}{3}(|\Delta_2|+|\Delta_3|+|\Delta_4|)+\frac{1}{2}|\Delta_5|+\frac{1}{6}|\Delta_7|\\
&=&|\Delta_2|+\frac{1}{2}|\Delta_5|+\frac{1}{6}\big(\big{|}\mathcal{X}\big|-3|\Delta_2|-|\Delta_5|\big)\\
&=&\frac{1}{6}\big(\big{|}\mathcal{X}\big|+3|\Delta_2|+2|\Delta_5|\big),
\end{eqnarray*}
which is our desired result. We are done.
\end{proof}
We are left to compute $|\Delta_2|$ and $|\Delta_5|$.
The following result exhibits  the value of $|\Delta_2|$ explicitly.
\begin{lem}\label{delta2}
%Let $r\geq3$ be a positive integer,  $A_2=\begin{pmatrix} 0 & 1 \\1 & 0\end{pmatrix}$ and
%let
%$$\Delta_2=\Big\{h(x)\in \mathcal{I}_r\,\Big|\, A_2h(x)=h(x)\hbox{~and~$h(x)$ divides $x^{2^r}+x$}\Big\}.$$
We have
\begin{equation*}
|\Delta_2|=
\begin{cases}
0,~~~~~~~~~~~~~~~~~~~\hbox{if~$r$~is odd,} \\
\vspace{0.1cm}\\
\frac{1}{r}\sum\limits_{d\mid \frac{r}{2}\atop{d ~odd}}\mu(d)2^{\frac{r}{2d}},~~\hbox{if~$r$~is even,}
\end{cases}
\end{equation*}
where $\mu$ is the M\"{o}bius function.
\end{lem}
\begin{proof}
It is easily seen that the order of $A_2\in{\rm PGL}$ is equal to two, i.e.,
$A_2^2=E_2$ ($E_2$ is the identity element of the group ${\rm PGL}$)  and $A_2\neq E_2$.
If $r$ is odd then $|\Delta_2|=0$; this is simply because there is no monic irreducible polynomial  $h(x)$ of odd degree   over $\mathbb{F}_q$
satisfying $A_2h(x)=h(x)$.

We assume that $r=2r'$, where $r'>1$ is a positive integer.
By the very definition of the group action of ${\rm PGL}$ on $\mathcal{I}_r$, we see that
$A_2h(x)$ is equal to the monic polynomial $A_2h(x)=\Big(x^rh(\frac{1}{x})\Big)^*$.
The irreducible polynomials   satisfying $\Big(x^rh(\frac{1}{x})\Big)^*=h(x)$ are termed as self-reciprocal irreducible monic
polynomials  in the literature, which have been   studied extensively.
It is known that $h(x)\in \mathcal{I}_r$ is self-reciprocal (equivalently $A_2h(x)=h(x)$)
if and only if $h(x)$ divides $x^{q^{r'}+1}+1$ (see \cite[Theorem 1]{Meyn} or one can prove this fact easily).
Therefore,  $A_2h(x)=h(x)$ and $h(x)$ divides $x^{2^r}+x$ if and only if $h(x)$ divides
$\gcd\big(x^{q^{r'}+1}+1,x^{2^r}+x\big)$.
We claim that
\begin{equation*}
\begin{split}
\gcd\big(x^{q^{r'}+1}+1,x^{2^r}+x\big)&=\gcd\big(x^{q^{r'}+1}+1,x^{2^r-1}+1\big)\\
&=x^{\gcd(q^{r'}+1,2^r-1)}+1\\
&=x^{\gcd(2^{nr'}+1,2^r-1)}+1\\
&=x^{2^{r'}+1}+1.
\end{split}
\end{equation*}
To prove the claim, it is enough to show that
\begin{equation}\label{gcd}
\gcd(2^{nr'}+1,2^r-1)=2^{r'}+1.
\end{equation}
Indeed,
observing that $2^{nr}-1=(2^{nr'}-1)(2^{nr'}+1)$,  $\gcd(2^{nr'}-1,2^{nr'}+1)=1$ and $2^{r}-1$ divides $2^{nr}-1$, we have
$$
2^r-1=\gcd(2^{nr}-1,2^r-1)=\gcd\Big(\big(2^{nr'}-1\big)\big(2^{nr'}+1\big),2^r-1\Big)
$$
and
\begin{equation*}
\begin{split}
\gcd\Big(\big(2^{nr'}-1\big)\big(2^{nr'}+1\big),2^r-1\Big)&=\gcd\big(2^{nr'}-1,2^r-1\big)\times\gcd\big(2^{nr'}+1,2^r-1\big)\\
&=\big(2^{\gcd(nr',r)}-1\big)\times\gcd\big(2^{nr'}+1,2^r-1\big)\\
&=\big(2^{r'}-1\big)\times\gcd\big(2^{nr'}+1,2^r-1\big),
\end{split}
\end{equation*}
where the last equality holds because $\gcd(n,r)=1$. Therefore, we have
$$
\gcd\big(2^{nr'}+1,2^r-1\big)=\frac{2^r-1}{2^{r'}-1}=2^{r'}+1.
$$
We have thus shown that
$\gcd\big(x^{q^{r'}+1}+1,x^{2^r}+x\big)=x^{2^{r'}+1}+1$.
This implies that the number of  monic irreducible polynomials $h(x)$ of degree $r$
over $\mathbb{F}_q$ that satisfy $A_2h(x)=h(x)$ and $h(x)\mid(x^{2^r}+x)$ (namely the size of $\Delta_2$) is equal to the
number of monic irreducible polynomials of degree $r$ over $\mathbb{F}_q$ that divide $x^{2^{r'}+1}+1$.

It is readily seen that every irreducible factor (except the one $x+1$) of $x^{2^{r'}+1}+1$ over $\mathbb{F}_q$ has even degree.
If $f(x)$ is a monic irreducible polynomial of degree $2d$ over $\mathbb{F}_q$ that divides $x^{2^{r'}+1}+1$,
we assert that $d$ is a divisor of $r'$ and $r'/d$ is odd.
To this end,
note, by $f(x)\mid(x^{2^{r'}+1}+1)$,  that  $f(x)\mid(x^{2^{r}-1}+1)$ and thus $f(x)\mid(x^{2^{nr}-1}+1)$, which yields
$f(x)\mid(x^{q^{r}-1}+1)$. Since $f(x)$ is irreducible of degree $2d$ over $\mathbb{F}_q$, we have  $(q^{2d}-1)\mid(q^{r}-1)$.
We then have that
$2d$ is a divisor of $r$, which implies that $d$ divides $r'$.
Since $n$ is an odd prime number, it follows that $x^{2^{r'}+1}+1$ divides
$x^{2^{nr'}+1}+1=x^{q^{r'}+1}+1$. Thus,  $f(x)$ divides $\gcd(x^{q^d+1}+1,x^{q^{r'}+1}+1)=x^{\gcd(q^d+1,q^{r'}+1)}+1$.
If $r'/d=r''$ is even, then $\gcd(q^d+1,q^{r'}+1)=1$; otherwise, let $\gcd(q^d+1,q^{r'}+1)=\ell\neq1$. Then
$q^d\equiv-1\pmod{\ell}$ and $(q^d)^{r''}\equiv(-1)^{r''}\pmod{\ell}$. We would have $q^{r'}\equiv1\pmod{\ell}$ since $r''$ is even,
contradicting to $q^{r'}\equiv-1\pmod{\ell}$. We thus have proven the assertion.

Let $R_{d}(x)$ be the product of all monic irreducible polynomials of degree $2d$ over $\mathbb{F}_q$ which divide $x^{2^{r'}+1}+1$,
in symbols
$$
R_{d}(x)=\prod\Big\{h(x)\,\Big|\,h(x)\in \mathcal{I}_{2d} ~\hbox{and $h(x)$ divides $x^{2^{r'}+1}+1$}\Big\}.
$$
It follows that
$$
x^{2^{r'}+1}+1=(x+1)\prod\limits_{d\mid r'\atop{r'/d ~odd}}R_{d}(x).
$$
Let $H_{r'}^0(x)=\frac{x^{2^{r'}+1}+1}{x+1}$. By the M\"{o}bius inversion formula (for example, see \cite[Theorem 3.24]{lidl}),
we have
$$
R_{r'}(x)=\prod\limits_{d\mid r'\atop{d ~odd}}H^0_{r'/d}(x)^{\mu(d)}.
$$
We conclude that
$$
r|\Delta_2|=\sum\limits_{d\mid \frac{r}{2}\atop{d ~odd}}\mu(d)2^{\frac{r}{2d}}.
$$
We are done.
\end{proof}

We need to find the value of $|\Delta_5|$. For this purpose, we first establish several lemmas.
\begin{lem}\label{equivalent2}
Let $f(x)\in \mathcal{I}_r$. Then $f(x)\in \mathcal{X}$ and
$A_5f(x)=f(x)$ if and only if
$$f(x)\Big|\gcd\big(x^{q^{r_0}+1}+x+1, x^{2^r}+x\big),$$
where $r_0=\frac{r}{3}$.% and is a positive integer.
\end{lem}
\begin{proof}
Suppose that $f(x)\in \mathcal{X}$ and $A_5f(x)=f(x)$, which gives
 $3\mid r$. Assume that $\alpha\in \mathbb{F}_{q^r}$ is a root of $f(x)$.
Then there exists a positive integer $r_0~(1\leq r_0\leq r-1)$ such that
$$\frac{\alpha+1}{\alpha}=A_5\alpha=\alpha^{q^{r_0}},$$
where $r_0$ is the least positive integer satisfying the above equality. Thus,
$f(x)\,\Big|\,\big(x^{q^{r_0}+1}+x+1\big),$
which implies that
$$f(x)\,\Big|\,\gcd\big(x^{q^{r_0}+1}+x+1, x^{2^r}+x\big).$$

Let $$\Omega=\big\{\alpha,\alpha^q, \cdots,\alpha^{q^{r-1}}\big\}$$
be the set of all the roots of $f(x)$. Then the cyclic group $\langle A_5\rangle$ acts on the set $\Omega$, and
 $\Omega$ can be decomposed into  disjoint union of orbits:
$$\Omega=\langle A_5\rangle(\alpha)\cup\langle A_5\rangle(\alpha^q)\cup\cdots \cup\langle A_5\rangle(\alpha^{q^{r_0-1}}).$$
Hence we get
$r=3r_0.$

On the contrary, let $\alpha\in \mathbb{F}_{q^r}$ be a root of $f(x)$. Since $f(x)\,\Big|\,\big(x^{q^{r_0}+1}+x+1, x^{2^r}+x\big)$,
we have $f(x)\in \mathcal{X}$ and
$A_5\alpha=\frac{\alpha+1}{\alpha}=\alpha^{q^{r_0}}.$
Hence,  $A_5\alpha$ is a root of $f(x)$ and $A_5f(x)=f(x)$.
\end{proof}

The next lemma is important in determining the value of $|\Delta_5|$.
\begin{lem}\label{factorization}
Let $f(x)\in \mathcal{X}$ with
$A_5f(x)=f(x)$ and let $\alpha\in \mathbb{F}_{q^r}$ be a root of $f(x)$. Then
 $$x^{q^{r_0}+1}+x+1=(x-\alpha)\prod_{\gamma\in \mathbb{F}_{q^{r_0}}}\Big(x-\big(\frac{\gamma^2+\gamma+1}{\alpha+\gamma}+\gamma+1\big)\Big),$$
where $r=3r_0$.
\end{lem}
\begin{proof}
According to Lemma \ref{equivalent2}, $\alpha$ is a root of $x^{q^{r_0}+1}+x+1$ and  $\alpha^{q^{r_0}}=\frac{\alpha+1}{\alpha}$.
Thus for any $\gamma\in \mathbb{F}_{q^{r_0}}$,
\begin{eqnarray*}
\Big(\frac{\gamma^2+\gamma+1}{\alpha+\gamma}+\gamma+1\Big)^{q^{r_0}}&=&\frac{\gamma^2+\gamma+1}{\alpha^{q^{r_0}}+\gamma}+\gamma+1\\
&=&\frac{\gamma^2+\gamma+1}{\frac{\alpha+1}{\alpha}+\gamma}+\gamma+1.
\end{eqnarray*}
We further obtain that
\begin{eqnarray*}
&&\Big(\frac{\gamma^2+\gamma+1}{\alpha+\gamma}+\gamma+1\Big)^{q^{r_0}+1}+\Big(\frac{\gamma^2+\gamma+1}{\alpha+\gamma}+\gamma+1\Big)+1\\
&=& \Big(\frac{\gamma^2+\gamma+1}{\alpha+\gamma}+\gamma+1\Big)\Big(\big(\frac{\gamma^2+\gamma+1}{\alpha+\gamma}+\gamma+1\big)^{q^{r_0}}+1\Big)+1\\
&=&\Big(\frac{\gamma^2+\gamma+1}{\alpha+\gamma}+\gamma+1\Big)\Big(\frac{\gamma^2+\gamma+1}{\frac{\alpha+1}{\alpha}+\gamma}+\gamma\Big)+1\\
&=&\frac{\alpha\gamma+\alpha+1}{\alpha+\gamma}\cdot\frac{\alpha+\gamma}{\alpha\gamma+\alpha+1}+1\\
&=&1+1\\
&=&0.
\end{eqnarray*}
Therefore $\alpha$ and $\frac{\gamma^2+\gamma+1}{\alpha+\gamma}+\gamma+1~$ $(\hbox{for any}~\gamma\in \mathbb{F}_{q^{r_0}})$ are the roots of
$x^{q^{r_0}+1}+x+1$.
Let $\theta_\gamma=\frac{\gamma^2+\gamma+1}{\alpha+\gamma}+\gamma+1$.
In the following we check that  for any $\gamma, \gamma'\in \mathbb{F}_{q^{r_0}}$,
$$\alpha\neq \theta_\gamma~\hbox{and}~\theta_\gamma\neq\theta_{\gamma'}~(\gamma\neq\gamma'),$$
which shows that the roots of $x^{q^{r_0}+1}+x+1$ are distinct. Indeed,
\begin{eqnarray*}
&&\theta_\gamma=\theta_{\gamma'}\\
&\Leftrightarrow & \frac{\gamma^2+\gamma+1}{\alpha+\gamma}+\gamma+1=\frac{\gamma'^2+\gamma'+1}{\alpha+\gamma'}+\gamma'+1\\
&\Leftrightarrow&\frac{\gamma^2+\gamma+1}{\alpha+\gamma}=\frac{\gamma'^2+\gamma'+1}{\alpha+\gamma'}+\gamma+\gamma'\\
&\Leftrightarrow& (\gamma^2+\gamma+1)(\alpha+\gamma')=(\gamma'^2+\gamma'+1)(\alpha+\gamma)+(\gamma+\gamma')(\alpha+\gamma)(\alpha+\gamma')\\
&\Leftrightarrow & \alpha\gamma+\gamma'=\alpha\gamma'+\gamma+\alpha^2\gamma+\alpha^2\gamma'\\
&\Leftrightarrow & (\gamma+\gamma')(\alpha^2+\alpha+1)=0\\
&\Leftrightarrow & \gamma=\gamma'.
\end{eqnarray*}
In the last equality but one $r\geq3$ implies  that $\alpha^2+\alpha+1\neq0$.
Next, let us verify  that for any $\gamma\in \mathbb{F}_{q^{r_0}}$,
\begin{eqnarray*}
&&\alpha=\theta_{\gamma'}\\
&\Leftrightarrow & \alpha=\frac{\gamma^2+\gamma+1}{\alpha+\gamma}+\gamma+1\\
&\Leftrightarrow&\alpha(\alpha+\gamma)=\gamma^2+\gamma+1+(\alpha+\gamma)(\gamma+1)\\
&\Leftrightarrow & \alpha^2+\alpha+1=0.
\end{eqnarray*}

In conclusion, we see that all the roots of $x^{q^{r_0}+1}+x+1$ are distinct.
Note that the degree of this polynomial is $q^{r_0}+1$.
We have decomposed  completely the polynomial $x^{q^{r_0}+1}+x+1$ into degree-one factors in $\mathbb{F}_{q^r}$.
\end{proof}

Given a positive integer $n',$
let $N_q(n')$ be the number of monic irreducible polynomials in $\mathbb{F}_q[x]$ of degree $n'$.
According to \cite[Theorem 3.25]{lidl},  the number $N_q(n')$ is given by
$$N_q(n')=\frac{1}{n'}\sum_{d|n'}\mu(d)q^{\frac{n'}{d}}.$$
Thus, a crude estimate yields
$$N_q(n')> 0.$$
In other words,  for every finite field $\mathbb{F}_q$ and every positive integer $n'$,  there
exists an irreducible polynomial in $\mathbb{F}_q[x]$ of degree $n'$.

In addition, the M\"{o}bius function $\mu$ satisfies (see \cite[Lemma 3.23]{lidl})
$$\sum_{d\mid n'}\mu(d)=
\begin{cases}
1 & \mbox{if}~n'=1,\\
0 & \mbox{if}~n'>1.
\end{cases}
$$

With these known results, we have the following result
which guarantees  the existence of a monic irreducible polynomial $f(x)$ of degree $r$ over $\mathbb{F}_2$
satisfying $A_5f(x)=f(x)$.

\begin{lem}\label{f2irre}
Let $r$ be an integer satisfying $3\mid r$ and $r\neq6$.
Then there exists a monic irreducible polynomial $f(x)$ of degree $r$ over $\mathbb{F}_2$
satisfying $A_5f(x)=f(x)$.
\end{lem}
\begin{proof}
According to \cite[Theorem 4.7]{Reis20}, the number $\mathcal{N}(\mathbb{F}_2,r)$ of  monic irreducible polynomials $f(x)$ of degree $r$ over $\mathbb{F}_2$
satisfying $A_5f(x)=f(x)$ is equal to
$$\mathcal{N}(\mathbb{F}_2,r)=\frac{2}{r} \sum_{d\mid\frac{r}{3}\atop \gcd(3,d)=1}\Big(2^{\frac{r}{3d}}+(-1)^{\frac{r}{3d}+1}\Big)\mu(d),$$
where $\mu$ is the M\"{o}bius function.
In the following we aim to prove that
$$\mathcal{N}_0(\mathbb{F}_2,r)=\sum_{d\mid\frac{r}{3}\atop \gcd(3,d)=1}\Big(2^{\frac{r}{3d}}+(-1)^{\frac{r}{3d}+1}\Big)\mu(d)>0.$$
To this end, suppose that $r=3m$ and $m=2^l3^kp_1^{l_1}p_2^{l_2}\cdots p_t^{l_t}$,
where $l,k, t$ are  non-negative integers  and  $p_i$ are  prime numbers with $p_i\neq2,3$ for $i=1,2,\cdots,t$.
Let $m_0=2^lp_1^{l_1}p_2^{l_2}\cdots p_t^{l_t}$. We consider two case separately.

(1) $k\geq1$.
In this case let $a=2^{3^k}$, then $a\geq 8$. Thus,
\begin{eqnarray*}
\mathcal{N}_0(\mathbb{F}_2,r)&=&\sum_{d|\frac{r}{3}\atop \gcd(3,d)=1}\Big(2^{\frac{r}{3d}}+(-1)^{\frac{r}{3d}+1}\Big)\mu(d)\\
&=&\sum_{d|m\atop \gcd(3,d)=1}\Big(2^{\frac{m}{d}}-(-1)^{\frac{m}{d}}\Big)\mu(d)\\
&=&\sum_{d|m_0}\Big(2^{\frac{3^km_0}{d}}-(-1)^{\frac{3^km_0}{d}}\Big)\mu(d)\\
&=&\sum_{d|m_0}\Big(a^{\frac{m_0}{d}}-(-1)^{\frac{m_0}{d}}\Big)\mu(d)\\
&\geq& (a^{m_0}-1)-(a^{m_0-1}+1)-\cdots-(a^2+1)-(a+1)\\
&=& a^{m_0}-a^{m_0-1}-\cdots-a^2-a-m_0\\
&=& a^{m_0}-\frac{a^{m_0}-a}{a-1}-m_0\\
&=&a^{m_0}\cdot \frac{a-2}{a-1}+\frac{a}{a-1}-m_0\\
&>& a^{m_0}\cdot \frac{a-2}{a-1}-(m_0-1)\\
&=& a^{m_0-1}\cdot \frac{a(a-2)}{a-1}-m_0\\
&>&a^{m_0-1}-(m_0-1)>1.
\end{eqnarray*}

(2) $k=0$. In this case $m=m_0=2^lp_1^{l_1}p_2^{l_2}\cdots p_t^{l_t}$ and we have
\begin{eqnarray*}
\mathcal{N}_0(\mathbb{F}_2,r)&=&\sum_{d|\frac{r}{3}\atop \gcd(3,d)=1}\Big(2^{\frac{r}{3d}}+(-1)^{\frac{r}{3d}+1}\Big)\mu(d)\\
&=&\sum_{d|m\atop \gcd(3,d)=1}\Big(2^{\frac{m}{d}}-(-1)^{\frac{m}{d}}\Big)\mu(d)\\
&=&\sum_{d|m_0}\Big(2^{\frac{m_0}{d}}-(-1)^{\frac{m_0}{d}}\Big)\mu(d).
\end{eqnarray*}

At this point, we need to calculate five subcases separately.

(2.1) $l=0, t=0$. In this subcase we have $m=m_0=1$. Thus
\begin{eqnarray*}
\mathcal{N}_0(\mathbb{F}_2,r)&=&\sum_{d|m_0}\Big(2^{\frac{m_0}{d}}-(-1)^{\frac{m_0}{d}}\Big)\mu(d)\\
&=&\sum_{d|m_0}\Big(2^{\frac{m_0}{d}}+1\Big)\mu(d)\\
&=&3>0.
\end{eqnarray*}

(2.2) $l=0, t\geq 1$. In this subcase we have $m_0>1$ and $2\nmid m_0$. Thus
\begin{eqnarray*}
\mathcal{N}_0(\mathbb{F}_2,r)&=&\sum_{d|m_0}\Big(2^{\frac{m_0}{d}}-(-1)^{\frac{m_0}{d}}\Big)\mu(d)\\
&=&\sum_{d|m_0}\Big(2^{\frac{m_0}{d}}+1\Big)\mu(d)\\
&=&\sum_{d|m_0}\mu(d)2^{\frac{m_0}{d}}+\sum_{d|m_0}\mu(d)\\
&=& N_2(m_0)+0=N_2(m_0)>0.
\end{eqnarray*}

(2.3) $l\geq2, t=0$. In this subcase,   $m_0=2^l$. Thus
\begin{eqnarray*}
\mathcal{N}_0(\mathbb{F}_2,r)&=&\sum_{d|m_0}\Big(2^{\frac{m_0}{d}}-(-1)^{\frac{m_0}{d}}\Big)\mu(d)\\
&=&\sum_{d|2^l}\Big(2^{\frac{2^l}{d}}-(-1)^{\frac{2^l}{d}}\Big)\mu(d)\\
&=&2^{l-1}\geq2.
\end{eqnarray*}

(2.4) $l=1, t\neq0$. In this subcase let $n_0=p_1^{l_1}p_2^{l_2}\cdots p_t^{l_t}$ and then $m_0=2n_0$ with $n_0\geq 5$ being odd.
It follows that
\begin{eqnarray*}
\mathcal{N}_0(\mathbb{F}_2,r)&=&\sum_{d|m_0}\Big(2^{\frac{m_0}{d}}-(-1)^{\frac{m_0}{d}}\Big)\mu(d)\\
&=&\sum_{d|n_0}\Big(2^{\frac{m_0}{d}}-(-1)^{\frac{m_0}{d}}\Big)\mu(d)+\sum_{d|n_0}\Big(2^{\frac{m_0}{2d}}-(-1)^{\frac{m_0}{2d}}\Big)\mu(2d)\\
&=&\sum_{d|n_0}\Big(2^{\frac{m_0}{d}}-1\Big)\mu(d)+\sum_{d|n_0}\Big(2^{\frac{m_0}{2d}}+1\Big)\mu(2d)\\
&=& \sum_{d|n_0}\mu(d)2^{\frac{m_0}{d}}+\sum_{d|n_0}\mu(2d)2^{\frac{m_0}{2d}}-\sum_{d|n_0}\mu(d)+\sum_{d|n_0}\mu(2d)\\
&=& \sum_{d|m_0}\mu(d)2^{\frac{m_0}{d}}-2\sum_{d|n_0}\mu(d)\\
&=& \sum_{d|m_0}\mu(d)2^{\frac{m_0}{d}}\\
&=& N_2(m_0)>0.
\end{eqnarray*}

(2.5) $l\geq2, t\neq0$. In this subcase let $n_0=p_1^{l_1}p_2^{l_2}\cdots p_t^{l_t}$ and then $m_0=2^ln_0$ with $n_0\geq 5$ being odd.
Thus
\begin{eqnarray*}
\mathcal{N}_0(\mathbb{F}_2,r)&=&\sum_{d|m_0}\Big(2^{\frac{m_0}{d}}-(-1)^{\frac{m_0}{d}}\Big)\mu(d)\\
&=&\sum_{d|n_0}\Big(2^{\frac{m_0}{d}}-(-1)^{\frac{m_0}{d}}\Big)\mu(d)+\sum_{d|n_0}\Big(2^{\frac{m_0}{2d}}-(-1)^{\frac{m_0}{2d}}\Big)\mu(2d)\\
&=&\sum_{d|n_0}\Big(2^{\frac{m_0}{d}}-1\Big)\mu(d)+\sum_{d|n_0}\Big(2^{\frac{m_0}{2d}}-1\Big)\mu(2d)\\
&=& \sum_{d|n_0}\mu(d)2^{\frac{m_0}{d}}+\sum_{d|n_0}\mu(2d)2^{\frac{m_0}{2d}}-\sum_{d|n_0}\mu(d)-\sum_{d|n_0}\mu(2d)\\
&=& \sum_{d|n_0}\mu(d)2^{\frac{m_0}{d}}+\sum_{d|n_0}\mu(2d)2^{\frac{m_0}{2d}}-\sum_{d|n_0}\mu(d)+\sum_{d|n_0}\mu(d)\\
&=& \sum_{d|m_0}\mu(d)2^{\frac{m_0}{d}}\\
&=& N_2(m_0)>0.
\end{eqnarray*}

In conclusion, we have that $\mathcal{N}_0(\mathbb{F}_2,r)>0$,
which shows that there exists a monic irreducible polynomial $f(x)$ of degree $r$ over $\mathbb{F}_2$
satisfying $A_5f(x)=f(x)$.
This completes the proof of Lemma \ref{f2irre}.
\end{proof}

Two remarks are in order at this point.
\begin{Remark}{\rm
Lemma \ref{f2irre} says that if $r$ is an integer with $3\mid r$ and $r\neq6$,
then there exists a monic irreducible polynomial $f(x)$ of degree $r$ over $\mathbb{F}_2$
satisfying $A_5f(x)=f(x)$. Since $\gcd(r,n)=1$ and $q=2^n$, by virtue of \cite[Corollary 3.47]{lidl},
we see that $f(x)$ remains irreducible over $\mathbb{F}_q$, namely $f(x)\in \mathcal{I}_r$.
Fix a root $\alpha$ of $f(x)$, thus $\alpha\in\mathbb{F}_{2^r}$ and naturally $\alpha^{2^r}=\alpha$.}
\end{Remark}

\begin{Remark}\label{r6f2}{\rm
Assume that $r=6$, i.e., $l=1, t=0$ in the proof of Lemma \ref{f2irre}. In this case we have $m_0=2$, and
\begin{eqnarray*}
\mathcal{N}_0(\mathbb{F}_2,r)&=&\sum_{d\mid m_0}\Big(2^{\frac{m_0}{d}}-(-1)^{\frac{m_0}{d}}\Big)\mu(d)\\
&=&\sum_{d\mid 2}\Big(2^{\frac{2}{d}}-(-1)^{\frac{2}{d}}\Big)\mu(d)\\
&=&0.
\end{eqnarray*}
That is to say, there  is indeed no    monic irreducible polynomials $f(x)$ of degree $6$ over $\mathbb{F}_2$
satisfying $A_5f(x)=f(x)$.}
\end{Remark}

\begin{lem}\label{decomposition}
Let $r=3r_0$ with $r_0\neq2$,
let $f(x)\in \mathbb{F}_2[x]$ be a monic irreducible polynomial of degree $r$ satisfying $A_5f(x)=f(x)$
and let $\alpha$ be a root of $f(x)$.
Suppose that $$F_{r_0}(x)=\gcd\big(x^{q^{r_0}+1}+x+1, x^{2^r}+x\big).$$
Then
\begin{itemize}
\item[(1)] $F_{r_0}(x)$ has no irreducible factor of degree $1$ over $\mathbb{F}_q$.

\item[(2)]  $F_{r_0}(x)$ has an irreducible factor of degree $2$ over $\mathbb{F}_q$ if and only if  $r_0$ is even; If this is the case,
   $x^2+x+1$ is a unique monic irreducible factor of $F_{r_0}(x)$ over $\mathbb{F}_q$ with degree $2$.

\item[(3)] $F_{r_0}(x)$ has the following decomposition over $\mathbb{F}_{2^r}$:
$$F_{r_0}(x)=\gcd\big(x^{q^{r_0}+1}+x+1, x^{2^r}+x\big)=(x-\alpha)\prod_{\gamma\in \mathbb{F}_{2^{r_0}}}\Big(x-\big(\frac{\gamma^2+\gamma+1}{\alpha+\gamma}+\gamma+1\big)\Big).$$
\end{itemize}
\end{lem}
\begin{proof}
(1) Suppose otherwise that $F_{r_0}(x)$ has an irreducible factor of degree $1$ over $\mathbb{F}_q$.
Then there exists an element $\beta\in \mathbb{F}_q$ such that $\beta^{2^r}+\beta=0$.
Using $\gcd(r, n)=1$,  we see that
$$\beta\in \mathbb{F}_q\cap \mathbb{F}_{2^r}=\mathbb{F}_{2^n}\cap \mathbb{F}_{2^r}=\mathbb{F}_{2^{\gcd(r,n)}}=\mathbb{F}_{2}.$$
However,  $\beta^{q^{r_0}+1}+\beta+1=0$, which is a contradiction.
Therefore $F_{r_0}(x)$ has no irreducible factor of degree $1$ over $\mathbb{F}_q$.

(2) First, suppose that $F_{r_0}(x)$ has an irreducible factor $x^2+ax+b$ of degree $2$ over $\mathbb{F}_q$,
where $a,b\in \mathbb{F}_q$. Since $(x^2+ax+b)\,\big|\,(x^{2^r}+x)$, we have $\mathbb{F}_{q^2}\subseteq \mathbb{F}_{2^r}$,
i.e., $\mathbb{F}_{2^{2n}}\subseteq \mathbb{F}_{2^r}$. This leads to $2n\mid r$, i.e., $2n\mid3r_0$. Then $r_0$ must be even.

Second, suppose that $r_0$ is even. It is easy to see that $x^2+x+1$ is irreducible over $\mathbb{F}_2$.
Since $\gcd(r,n)=1$, $x^2+x+1$ is also irreducible over $\mathbb{F}_q$. In the following,  we aim to  show that $(x^2+x+1)\,\big|\,F_{r_0}(x)$.
Let $r_0=2r_1$. Thus $r=6r_1$ and further $3\,\big|\,(2^{6r_1}-1)$. Then $x^3-1\,\big|\,\big(x^{2^{6r_1}}-1\big)$.
Since $x^{2^r}-x=x\big(x^{2^{6r_1}}-1\big)$ and $x^3-1=(x-1)(x^2+x+1)$, we obtain  $(x^2+x+1)\,\big|\,(x^{2^r}-x)$.
Assume that $\vartheta,\vartheta^2$ are all the roots of $x^2+x+1$. Then $\vartheta^2=\vartheta+1$. Thus,
\begin{eqnarray*}
\vartheta^{2^l}=
\begin{cases}
\vartheta, & 2\mid l,\\
\vartheta+1, & 2\nmid l.
\end{cases}
\end{eqnarray*}
Therefore,
$$\vartheta^{q^{r_0}+1}+\vartheta+1=\vartheta\cdot\vartheta^{2^{nr_0}}+\vartheta+1=\vartheta^2+\vartheta+1=0$$
and
$$(\vartheta^2)^{q^{r_0}+1}+\vartheta^2+1=\vartheta^2\cdot(\vartheta^2)^{2^{nr_0}}+\vartheta^2+1
=\vartheta^2\cdot\vartheta^{2^{nr_0+1}}+\vartheta^2+1=\vartheta^2(\vartheta+1)+\vartheta^2+1=0.$$
We have shown that  $(x^2+x+1)\,\big|\,F_{r_0}(x)$.
It needs to shows that $x^2+x+1$ is a unique monic irreducible divisor of $F_{r_0}(x)$ over $\mathbb{F}_q$.
For this purpose, suppose $\alpha$ is a root of
$F_{r_0}(x)$ satisfying $\alpha^{q^2}=\alpha$ and $\alpha\neq\alpha^q$.
Since $F_{r_0}(x)$ is a divisor of $x^{q^{r_0}+1}+x+1$ and $r_0$ is even, then  $\alpha^{q^{r_0}+1}+\alpha+1=0$.
We thus have $\alpha^{2}+\alpha+1=0$, implying that $\alpha$ is a root of $x^2+x+1$.

(3) Note, by Lemma \ref{factorization},  that
$$x^{q^{r_0}+1}+x+1=(x-\alpha)\prod_{\gamma\in \mathbb{F}_{q^{r_0}}}\Big(x-\big(\frac{\gamma^2+\gamma+1}{\alpha+\gamma}+\gamma+1\big)\Big).$$
Since $\alpha^{2^r}=\alpha$, we have  $(x-\alpha)\,\big|\,F_{r_0}(x)$. On the other hand,
for any $\theta_\gamma=\frac{\gamma^2+\gamma+1}{\alpha+\gamma}+\gamma+1$,
\begin{eqnarray*}
&&\theta_\gamma^{2^r}=\theta_\gamma \\
&\Leftrightarrow & \Big(\frac{\gamma^2+\gamma+1}{\alpha+\gamma}+\gamma+1\Big)^{2^r}=\frac{\gamma^2+\gamma+1}{\alpha+\gamma}+\gamma+1\\
&\Leftrightarrow & \frac{\gamma^{2^{r+1}}+\gamma^{2^r}+1}{\alpha^{2^r}+\gamma^{2^r}}+\gamma^{2^r}=\frac{\gamma^2+\gamma+1}{\alpha+\gamma}+\gamma\\
&\Leftrightarrow &\frac{\gamma^{2^{r+1}}+\gamma^{2^r}+1}{\alpha^{2^r}+\gamma^{2^r}}+\frac{\gamma^2+\gamma+1}{\alpha+\gamma}=\gamma^{2^r}+\gamma\\
&\Leftrightarrow & \big(\gamma^{2^{r+1}}+\gamma^{2^r}+1\big)(\alpha+\gamma)+(\gamma^2+\gamma+1)\big(\alpha^{2^r}+\gamma^{2^r}\big)
      =(\gamma^{2^r}+\gamma)(\alpha+\gamma)(\alpha^{2^r}+\gamma^{2^r})\\
&\Leftrightarrow &\big(\gamma^{2^{r+1}}+\gamma^{2^r}+1\big)(\alpha+\gamma)+(\gamma^2+\gamma+1)\big(\alpha+\gamma^{2^r}\big)
      =(\gamma^{2^r}+\gamma)(\alpha+\gamma)(\alpha+\gamma^{2^r})\\
&\Leftrightarrow & (\gamma^{2^r}+\gamma)(\alpha+1)=\alpha^2(\gamma^{2^r}+\gamma)\\
&\Leftrightarrow & \gamma^{2^r}=\gamma\\
&\Leftrightarrow & \gamma\in \mathbb{F}_{2^r}.
\end{eqnarray*}
Hence we have
$$\gamma\in \mathbb{F}_{2^r}\cap \mathbb{F}_{q^{r_0}}=\mathbb{F}_{2^r}\cap \mathbb{F}_{2^{nr_0}}=\mathbb{F}_{2^{r_0}}.$$
Therefore,
$$F_{r_0}(x)=\gcd\big(x^{q^{r_0}+1}+x+1, x^{2^r}+x\big)=(x-\alpha)\prod_{\gamma\in \mathbb{F}_{2^{r_0}}}\Big(x-\big(\frac{\gamma^2+\gamma+1}{\alpha+\gamma}+\gamma+1\big)\Big).$$
\end{proof}

We are now in a position to give an  enumerative formula for the size of the set $\Delta_5$, which is based on the M$\ddot{o}$bius inversion
formula and its generalizations, see \cite[Proposition 5.2]{knop}.

\begin{lem}\label{mobius}
Let $\chi: \mathbb{N}\rightarrow \mathbb{C}$ be a completely multiplicative function, which is, in other words,
a homomorphism between the monoids $(\mathbb{N},+)$ and $(\mathbb{C},\cdot)$.
Let $\mathcal{F}, \mathcal{G}: \mathbb{N}\rightarrow \mathbb{C}$ be two functions such that
$$\mathcal{F}(n)=\sum_{d\mid n}\chi(d)\cdot \mathcal{G}(\frac{n}{d}), ~n\in \mathbb{N}.$$
Then,
$$\mathcal{G}(n)=\sum_{d\mid n}\chi(d)\cdot \mu(d)\cdot\mathcal{F}(\frac{n}{d}), ~n\in \mathbb{N}.$$
\end{lem}

\begin{lem}\label{delta5}
%Let $r\geq3$ be a positive integer,  $A_5=\begin{pmatrix} 1 & 1 \\1 & 0\end{pmatrix}$ and
%let
%$$\Delta_5=\Big\{h(x)\in \mathcal{I}_r\,\Big|\, A_5h(x)=h(x)\hbox{~and~$h(x)$ divides $x^{2^r}+x$}\Big\}.$$
We have
\begin{equation*}
|\Delta_5|=
\begin{cases}
0,~~~~~~~~~~~~~~~~~~~\hbox{$3\nmid r$,} \\
0,~~~~~~~~~~~~~~~~~~~\hbox{$r=6$,} \\
%\vspace{0.1cm}\\
\frac{1}{r}\sum\limits_{d\mid \frac{r}{3}\atop{\gcd(3,d)=1}}\mu(d)\Big(2^{\frac{r}{3d}}+(-1)^{\frac{r}{3d}+1}\Big),~~\hbox{$3\mid r~ and ~r\neq 6$,}
\end{cases}
\end{equation*}
where $\mu$ is the M\"{o}bius function.
\end{lem}
\begin{proof}
If  $3\nmid r$, it is easy to see that $|\Delta_5|=0$. If $r=6$,
we  consider the irreducible decomposition of $x^{2^r-1}-1$ over $\mathbb{F}_q$ and over $\mathbb{F}_2$, respectively.
Denote by $\mathbb{Z}^*_{2^r-1}$   the unit group of the residue   ring of integers modulo $2^r-1$.
Since $\gcd(q, 2^r-1)=1$, we have $q\in \mathbb{Z}^*_{2^r-1}$. Let $\langle q\rangle$ denote the cyclic subgroup of $\mathbb{Z}^*_{2^r-1}$.
Then there is an action of the group $\langle q\rangle$ on the set $\mathbb{Z}_{2^r-1}$ given as follows:
\begin{eqnarray*}
\langle q\rangle \times \mathbb{Z}_{2^r-1} &\longrightarrow & \mathbb{Z}_{2^r-1} \\
(q^i, ~k) &\mapsto & q^ik.
\end{eqnarray*}
Let $s$ be an integer with $0\leq s< n$.
Thus the $q$-cyclotomic coset of $s$ modulo $2^r-1$ is the same as the orbit
$\langle q\rangle s=\{q^is\,|\,i~\hbox{is an integer}\}$ of $s$ under this
group action. For the same reason,  there exists an action of the cyclic subgroup $\langle 2\rangle$ on the set $\mathbb{Z}_{2^r-1}$,
and the $2$-cyclotomic coset of $s$ modulo $2^r-1$ is the same as the orbit $\langle 2\rangle s$ of $s$ under this
group action. Noting that $\langle q\rangle\subseteq \langle 2\rangle$ and the order ${\rm ord}_{2^r-1}(q)$ of $q$ modulo $2^r-1$ is
$${\rm ord}_{2^r-1}(q)=\frac{{\rm ord}_{2^r-1}(2)}{\gcd(n,{\rm ord}_{2^r-1}(2))}=\frac{r}{\gcd(n,r)}=r={\rm ord}_{2^r-1}(2),$$
we obtain that $\langle q\rangle s=\langle 2\rangle s$ for any $s$ with $0\leq s< n$.
It follows that  the irreducible decompositions of $x^{2^r-1}-1$ over $\mathbb{F}_q$ and over $\mathbb{F}_2$ are the same.
By Remark \ref{r6f2},  there is no   monic irreducible polynomial $f(x)$ of degree $6$ over $\mathbb{F}_2$
satisfying $A_5f(x)=f(x)$. Hence, there is no   monic irreducible polynomial $f(x)$ of degree $6$ over $\mathbb{F}_q$
satisfying $A_5f(x)=f(x)$. Therefore in this case  we also have  $|\Delta_5|=0$.

In the following we consider the case where  $3\mid r$ and $r\neq 6$.
Assume that $r=3r_0$, where $r_0\neq2$ is a positive integer.
We have shown in Lemma \ref{decomposition} that
$$F_{r_0}(x)=\gcd\big(x^{q^{r_0}+1}+x+1, x^{2^r}+x\big)=(x-\alpha)\prod_{\gamma\in \mathbb{F}_{2^{r_0}}}
\Big(x-\big(\frac{\gamma^2+\gamma+1}{\alpha+\gamma}+\gamma+1\big)\Big).$$
Let $f(x)$ be a monic irreducible polynomial of degree $3d$ over $\mathbb{F}_q$ dividing $F_{r_0}(x)$.
We assert that $d\mid r_0$ and $\gcd(3, \frac{r_0}{d})=1$.
To this end,
note, by $f(x)\mid(x^{2^r}+x)$,  that  $f(x)\mid(x^{2^{r}-1}+1)$ and thus $f(x)\mid(x^{2^{nr}-1}+1)$, which yields
$f(x)\mid(x^{q^{r}-1}+1)$ and so $f(x)\mid(x^{q^{r}}-x)$. Since $f(x)$ is irreducible of degree $3d$ over $\mathbb{F}_q$,
$3d$ is a divisor of $r$, which implies that $d$ divides $r_0$.
It remains to show that $\gcd(3, \frac{r_0}{d})=1$. Suppose otherwise that $\gcd(3, \frac{r_0}{d})>1$, i.e., $3\mid \frac{r_0}{d}$,
say $\frac{r_0}{d}=3\kappa$ for some integer $\kappa\geq1$. Let $\alpha$ be a root of $f(x)$. Since $f(x)\,\big|\,\big(x^{q^{r_0}+1}+x+1\big)$,
we have
$$\frac{\alpha+1}{\alpha}=\alpha^{q^{r_0}}.$$
The degree of $f(x)$ is $3d$, which implies that $f(x)\,\big|\,(x^{q^{3d}}-x)$, yielding $\alpha^{q^{3d}}=\alpha$ and
$$\alpha^{q^{r_0}}=\alpha^{q^{d\cdot \frac{r_0}{d}}}=\alpha^{q^{3d\kappa}}=\alpha.$$
Hence
$$\frac{\alpha+1}{\alpha}=\alpha.$$
This is an equation $\alpha^2+\alpha+1=0$ for $\alpha$ over $\mathbb{F}_q$ of degree $2$, which contradicts to the assumption $r\geq3$.
We thus have proven the assertion.
On the other hand, we have shown in Lemma \ref{decomposition} that
$$
\begin{cases}
(x^2+x+1)\,\big|\,F_{r_0}(x), & 2\mid r_0,\\
(x^2+x+1)\nmid F_{r_0}(x), & 2\nmid r_0.
\end{cases}
$$
Set $\varepsilon_{r_0}(x)=\gcd\big(x^2+x+1, F_{r_0}(x)\big)$ and
\begin{equation*}\chi_s(t)=
\begin{cases}
1, & \gcd(s,t)=1,\\
0, & \mbox{otherwise.}
\end{cases}
\end{equation*}
Let $R_{d}(x)$ be the product of all monic irreducible polynomials of degree $3d$ over $\mathbb{F}_q$
which divide $F_{r_0}(x)$,
in symbols
$$
R_{d}(x)=\prod\Big\{h(x)\,\Big|\,h(x)\in \mathcal{I}_{3d} ~\hbox{and $h(x)$ divides $F_{r_0}(x)$}\Big\}.
$$
It follows from Lemma \ref{decomposition} that
$$
\frac{F_{r_0}(x)}{\varepsilon_{r_0}(x)}=\prod\limits_{d\mid r_0\atop{\gcd(3,r_0/d)=1}}R_{d}(x).
$$
Note that $F_{r_0}(x)$ has degree $2^{r_0}+1$ and the degree of $\varepsilon_{r_0}(x)$ is either $0$ if $r_0$ is odd, or $2$ if $r_0$ is even.
Then, if we set $\epsilon(r_0)=(-1)^{r_0+1}$, then
$$2^{r_0}+\epsilon(r_0)=\sum\limits_{d\mid r_0\atop{\gcd(3,r_0/d)=1}}3d\big|R_{d}(x)\big|
=\sum\limits_{d\mid r_0}3d\big|R_{d}(x)\big|\cdot \chi_3\big(\frac{r_0}{d}\big).$$
By Lemma \ref{mobius},
$$3r_0\big|R_{r_0}(x)\big|=\sum\limits_{d\mid r_0}\chi_3(d)\mu(d)\Big(2^{\frac{r_0}{d}}+\epsilon(\frac{r_0}{d})\Big).$$
We then have
$$
r|\Delta_5|=\sum\limits_{d\mid \frac{r}{3}\atop{\gcd(3,d)=1}}\mu(d)\Big(2^{\frac{r}{3d}}+(-1)^{\frac{r}{3d}+1}\Big),
$$
which implies that
$$
|\Delta_5|=\frac{1}{r}\sum\limits_{d\mid \frac{r}{3}\atop{\gcd(3,d)=1}}\mu(d)\Big(2^{\frac{r}{3d}}+(-1)^{\frac{r}{3d}+1}\Big).
$$
We are done.
\end{proof}

\subsection{The number of orbits of ${\rm Gal}$ on ${\rm PGL}\backslash\mathcal{I}_r$}
Collecting all the results that we have established, we arrive at the following result, which gives
the number of orbits of ${\rm Gal}$ on ${\rm PGL}\backslash\mathcal{I}_r$ (or equivalently, the number of orbits
of ${\rm P\Gamma L}$ on $\mathcal{I}_r$).
\begin{Theorem}\label{theorem}
We assume that  $n\geq 5$ is an odd prime number,
$q=2^n$ and  $r\geq3$ is a positive integer satisfying $\gcd(r,n)=1$.
The number of orbits of ${\rm Gal}$ on ${\rm PGL}\backslash\mathcal{I}_r$ is given by
$$\frac{n-1}{6n}\big(|\mathcal{X}|+3|\Delta_2|+2|\Delta_5|\big)+\frac{1}{nq(q^2-1)}\big(\mathcal{N}_0+\mathcal{N}_1+\mathcal{N}_2+\mathcal{N}_3\big),$$
where the values of $|\mathcal{X}|,$   $|\Delta_2|$ and $|\Delta_5|$ were explicitly given in Lemmas \ref{firstnumber}, \ref{delta2}
and \ref{delta5} respectively, and the values of $\mathcal{N}_i$ for $0\leq i\leq 3$ were explicitly  determined in Theorem \ref{orbitsize1}.
%where $\varphi$ is the Euler's Totient function, $E(r,q)$ was defined as in (\ref{erq}) and $\mu$ is the M\"{o}bius function.
\end{Theorem}
\begin{proof}
%By Lemma \ref{important}, the number of inequivalent  extended irreducible binary Goppa codes of length $q+1$
%and degree $r$ is less than or equal to the number of orbits of ${\rm P\Gamma L}$ on $\mathcal{S}$.
%By Lemma \ref{orbit},
Recall, from    Theorem \ref{szero}, that $s_0$ denotes the number of orbits of ${\rm P\Gamma L}$ on $\mathcal{I}_r$ with size $1$.
Let $s$ be the number of orbits of ${\rm P\Gamma L}$ on $\mathcal{I}_r$. Then% Using Lemmas \ref{nexttolastlem} and \ref{lastlem}, we have
$$s_0+n\big(s-s_0\big)
=\big|{\rm PGL}\verb|\|\mathcal{I}_r\big|.$$
Substituting $s_0$ by $\frac{1}{6}\big(|\mathcal{X}|+3|\Delta_2|+2|\Delta_5|\big)$,  we have
$$\frac{1}{6}\big(|\mathcal{X}|+3|\Delta_2|+2|\Delta_5|\big)+n\Big(s-\frac{1}{6}\big(|\mathcal{X}|+3|\Delta_2|+2|\Delta_5|\big)\Big)
=\big|{\rm PGL}\verb|\|\mathcal{I}_r\big|,$$
from which we obtain
\begin{equation*}
\begin{split}
s&=\frac{n-1}{6n}\big(|\mathcal{X}|+3|\Delta_2|+2|\Delta_5|\big)+\frac{1}{nq(q^2-1)}\big(\mathcal{N}_0+\mathcal{N}_1+\mathcal{N}_2+\mathcal{N}_3\big).\\
\end{split}
\end{equation*}
%Substituting  $|\mathcal{I}_r|=\frac{1}{r}\sum\limits_{d\,|\,r}\mu(d)q^{r/d}$ into the above equation, we obtain the desired result.
We are done.
\end{proof}

\subsection{An upper bound for the number of extended Goppa codes}
By Lemma \ref{important}, the number of inequivalent  extended irreducible binary Goppa codes of length $q+1$
and degree $r$ is less than or equal to the number of orbits of ${\rm P\Gamma L}$ on $\mathcal{S}$.
Lemma \ref{orbit} tells us that the number of orbits of ${\rm P\Gamma L}$ on $\mathcal{S}$  is equal to the number of
orbits of ${\rm P\Gamma L}$ on $\mathcal{I}_r$.
%Collecting all the results that we have established, we arrive at the following result, which gives an upper bound for the
%number of inequivalent  extended irreducible binary Goppa codes of length $2^n+1$ and degree $r$.
%By Lemma \ref{orbit},
%let $s$ be the number of orbits of ${\rm P\Gamma L}$ on $\mathcal{I}_r$.
Lemma \ref{book} says that the number of
orbits of ${\rm P\Gamma L}$ on $\mathcal{I}_r$ is equal to
the number of orbits of ${\rm Gal}$ on ${\rm PGL}\backslash\mathcal{I}_r$. With Theorem \ref{theorem}
at hand, we immediately have the following result.
\begin{Theorem}\label{corollary}
We assume that  $n\geq 5$ is an odd prime number,
$q=2^n$ and  $r\geq3$ is a positive integer satisfying $\gcd(r,n)=1$.
The number of inequivalent extended irreducible binary Goppa codes of length $q+1$ and degree $r$ is at most
$$\frac{n-1}{6n}\big(|\mathcal{X}|+3|\Delta_2|+2|\Delta_5|\big)+\frac{1}{nq(q^2-1)}\big(\mathcal{N}_0+\mathcal{N}_1+\mathcal{N}_2+\mathcal{N}_3\big),$$
where the values of $|\mathcal{X}|,$   $|\Delta_2|$ and $|\Delta_5|$ were explicitly given in Lemmas \ref{firstnumber}, \ref{delta2}
and \ref{delta5} respectively, and the values of $\mathcal{N}_i$ for $0\leq i\leq 3$ were explicitly  determined in Theorem \ref{orbitsize1}.
\end{Theorem}

%We give a small example to illustrate Theorem \ref{theorem}.
%\begin{Example}{\rm
%Take $n=3$ and $r=5$ in Theorem \ref{theorem}. This gives $q=2^n=2^3=8$ and thus $q(q^2-1)=504$. It is readily seen that
%$\gcd(r,n)=\gcd(5,3)=1$ and $\gcd(r,q(q^2-1))=\gcd(5,504)=1$, namely, the conditions listed in Theorem \ref{theorem} are satisfied.
%One easily has that ${\rm E}(r,q)$ has a single element $31$, in symbols ${\rm E}(r,q)={\rm E}(5,8)=\{31\}$.
%Then Theorem \ref{theorem} says that the number of inequivalent extended irreducible binary Goppa codes of length $9$ and degree $5$ is at most
%$$\frac{n-1}{6rn}\sum_{e\in {\rm E}(r,q)}\phi(e)+\frac{\sum_{d\,|\,r}\mu(d)q^{r/d}}{rnq(q^2-1)}
%=\frac{2}{3}+\frac{8^5-8}{3\times5\times504}=5.$$
%}
%\end{Example}

\section{Corollaries of Theorem \ref{corollary}}
In this section, we apply Theorem \ref{corollary} to some special cases, including $r=4$, $2p$ ($p\geq3$ is a prime number) and $\gcd(r,q^3-q)=1$.
Some previously known results in the literature  are reobtained directly. Consequently, our main result, Theorem \ref{corollary},
naturally contains the main results of \cite{cz}, \cite{yueqin} and \cite{ryan15}.

\subsection{The case: $r=4$.}
We first apply Theorem \ref{corollary} to reobtain the main result of \cite{ryan15},
which established  an upper bound on the number of extended irreducible binary quartic Goppa codes of length $2^n+1$
(where $n>3$ is a prime number).
By Theorem \ref{orbitsize1} and simple computations, we have $\mathcal{N}_2=\mathcal{N}_3=0$,
$$
\mathcal{N}_0=\frac{1}{4}\sum_{d|4}\mu(d)q^{\frac{4}{d}}=\frac{1}{4}q^2(q^2-1)
$$
and
$$
\mathcal{N}_1=\frac{q^2-1}{4}\sum_{d|2 \atop \gcd(2,d)=1}\mu(d)q^{\frac{4}{2d}}=\frac{q^2-1}{4}\cdot q^2=\frac{1}{4}q^2(q^2-1).
$$
The number of orbits of  ${\rm PGL}$ on   $\mathcal{I}_r$ is equal to
\begin{eqnarray*}
\big|{\rm PGL}\verb|\|\mathcal{I}_r\big|&=&\frac{1}{q(q^2-1)}\big(\mathcal{N}_0+\mathcal{N}_1+\mathcal{N}_2+\mathcal{N}_3\big)\\
&=&\frac{1}{q(q^2-1)}\cdot \big(\frac{1}{4}q^2(q^2-1)+\frac{1}{4}q^2(q^2-1)+0+0\big)\\
&=&\frac{q}{2}.
\end{eqnarray*}
Next, by Lemmas \ref{firstnumber}, \ref{delta2}
and \ref{delta5}, we have $|\Delta_5|=0$,
$$
|\mathcal{X}|=\frac{1}{4}\sum_{d|4}\mu(d)(2^{\frac{4}{d}}-1)=3
$$
and
$$
|\Delta_2|=\frac{1}{r}\sum_{d|\frac{r}{2} \atop d~odd}\mu(d)2^{\frac{r}{2d}}=\frac{1}{4}\sum_{d|2 \atop d~odd}\mu(d)2^{\frac{2}{d}}=1.
$$
Thus the number $s_0$  ($s_0$ denotes the number of orbits of ${\rm P\Gamma L}$ on $\mathcal{I}_r$ with size $1$, see Theorem \ref{szero})
is equal to
$$s_0=\frac{1}{6}\big(|\mathcal{X}|+3|\Delta_2|+2|\Delta_5|\big)=1.$$
Let $s$ be the number of orbits of ${\rm P\Gamma L}$ on $\mathcal{I}_r$.
Then
$$1+n(s-1)=\big|{\rm PGL}\verb|\|\mathcal{I}_r\big|=\frac{q}{2},$$
leading to
$$s=\frac{1}{n}(\frac{q}{2}-1)+1=\frac{2^{n-1}-1}{n}+1.$$
As a corollary of Theorem \ref{corollary},  we have reobtained  the main result of \cite{ryan15}.
\begin{Corollary}
{\rm(\cite[Theorem 5.1]{ryan15})}  Let $n>3$ be a prime number. The number of extended irreducible
binary quartic Goppa codes of length $2^n+1$ is at most $\frac{2^{n-1}-1}{n}+1$.
\end{Corollary}

\subsection{The case: $r=2p$, $p\geq 3$ is a prime number}
We now turn to consider the case $r=2p$, where $p\geq3$ is a prime number. The particular case $p=3$ was considered in \cite{yueqin}.
We need to divide the case into three subcases separately: $p\mid(q-1)$, $p\mid(q+1)$ and the rest.

$\bullet$ Subcase 1: $p\mid(q-1)$.
In this subcase we must have $p\nmid(q+1)$. Using   Theorem \ref{orbitsize1} directly,  we have
$$
\mathcal{N}_0=\frac{1}{2p}\sum_{d|2p}\mu(d)q^{\frac{2p}{d}}=\frac{1}{2p}(q^{2p}-q^p-q^2+q),~~
\mathcal{N}_1=\frac{q^2-1}{2p}\sum_{d|p \atop \gcd(2,d)=1}\mu(d)q^{\frac{p}{d}}=\frac{q^2-1}{2p}(q^p-q),
$$
\begin{eqnarray*}
\mathcal{N}_2&=&q(q+1)\cdot\sum_{D|\gcd(2p,q-1) \atop D\neq 1}\frac{\varphi^2(D)}{2p}\sum_{d|\frac{2p}{D} \atop \gcd(d, D)=1}\mu(d)\big(q^{\frac{2p}{Dd}}-1\big)\\
&=&q(q+1)\cdot\sum_{D|p \atop D\neq 1}\frac{\varphi^2(D)}{2p}\sum_{d|\frac{2p}{D} \atop \gcd(d, D)=1}\mu(d)\big(q^{\frac{2p}{Dd}}-1\big)\\
&=& q(q+1)\cdot\frac{\varphi^2(p)}{2p}\sum_{d|2 \atop \gcd(d, p)=1}\mu(d)\big(q^{\frac{2}{d}}-1\big)\\
&=& q(q+1)\cdot\frac{(p-1)^2}{2p}\big(q^2-1-(q-1)\big)\\
&=& \frac{(p-1)^2}{2p}\cdot q^2(q^2-1),\\
\mathcal{N}_3&=&\frac{q(q-1)}{2}\cdot\sum_{D|\gcd(2p,q+1) \atop D\neq 1}\frac{\varphi^2(D)}{2p}\sum_{d|\frac{2p}{D} \atop \gcd(d, D)=1}\mu(d)\big(q^{\frac{2p}{Dd}}+(-1)^{\frac{2p}{Dd}+1}\big)=0.
\end{eqnarray*}
The number of orbits of ${\rm PGL}$ on   $\mathcal{I}_r$   is equal to
\begin{eqnarray*}
\big|{\rm PGL}\verb|\|\mathcal{I}_r\big|&=&\frac{1}{q(q^2-1)}\big(\mathcal{N}_0+\mathcal{N}_1+\mathcal{N}_2+\mathcal{N}_3\big)\\
&=&\frac{1}{q(q^2-1)}\Big(\frac{1}{2p}\big(q^{2p}-q^p-q^2+q+(q^2-1)(q^p-q)\big)+\frac{(p-1)^2}{2p}\cdot q^2(q^2-1)\Big)\\
&=&\frac{1}{2pq(q^2-1)}\big(q^{2p}+q^{p+2}-2q^p-q^3-q^2+2q\big)+\frac{q(p-1)^2}{2p}\\
&=&\frac{1}{2p(q^2-1)}\big(q^{2p-1}+q^{p+1}-2q^{p-1}-q^2-q+2\big)+\frac{q(p-1)^2}{2p}.
\end{eqnarray*}
Next, we have $|\Delta_5|=0$,
\begin{eqnarray*}
|\mathcal{X}|&=&\frac{1}{2p}\sum_{d|2p}\mu(d)(2^{\frac{2p}{d}}-1)=\frac{1}{2p}(2^{2p}-2^p-2^2+2)=\frac{1}{2p}(2^{2p}-2^p-2),\\
|\Delta_2|&=&\frac{1}{r}\sum_{d|\frac{r}{2} \atop d~odd}\mu(d)2^{\frac{r}{2d}}=\frac{1}{2p}\sum_{d|p \atop d~odd}\mu(d)2^{\frac{p}{d}}=\frac{1}{2p}(2^p-2).\\
\end{eqnarray*}
Thus the number $s_0$ of orbits of ${\rm P\Gamma L}$ on $\mathcal{I}_r$ with size $1$ is
$$s_0=\frac{1}{6}\big(|\mathcal{X}|+3|\Delta_2|+2|\Delta_5|\big)=\frac{1}{12p}(2^{2p}+2^{p+1}-8).$$
Let $s$ be the number of orbits of ${\rm P\Gamma L}$ on $\mathcal{I}_r$.
Then
$$s_0+n(s-s_0)=\big|{\rm PGL}\verb|\|\mathcal{I}_r\big|,$$
yielding
\begin{eqnarray*}
&&\frac{1}{12p}(2^{2p}+2^{p+1}-8)+n(s-s_0)=\frac{1}{2p(q^2-1)}\big(q^{2p-1}+q^{p+1}-2q^{p-1}-q^2-q+2\big)+\frac{q(p-1)^2}{2p}\\
&\Rightarrow& 2^{2p}+2^{p+1}-8+12pn(s-s_0)=\frac{6}{q^2-1}\big(q^{2p-1}+q^{p+1}-2q^{p-1}-q^2-q+2\big)+6q(p-1)^2\\
&\Rightarrow& 12pn(s-s_0)=\frac{6}{q^2-1}\big(q^{2p-1}+q^{p+1}-2q^{p-1}-q^2-q+2\big)+6q(p-1)^2-2^{2p}-2^{p+1}+8\\
&\Rightarrow& s-s_0=\frac{1}{2pn(q^2-1)}\big(q^{2p-1}+q^{p+1}-2q^{p-1}-q^2-q+2\big)+\frac{q(p-1)^2}{2pn}-\frac{2^{2p}+2^{p+1}-8}{12pn}\\
&\Rightarrow& s=\frac{1}{2pn(q^2-1)}\big(q^{2p-1}+q^{p+1}-2q^{p-1}-q^2-q+2\big)+\frac{q(p-1)^2}{2pn}-\frac{2^{2p}+2^{p+1}-8}{12pn}+s_0\\
&\Rightarrow& s=\frac{1}{2pn(q^2-1)}\big(q^{2p-1}+q^{p+1}-2q^{p-1}-q^2-q+2\big)+\frac{q(p-1)^2}{2pn}+\frac{1}{12p}(1-\frac{1}{n})(2^{2p}+2^{p+1}-8)\\
&\Rightarrow& s=\frac{1}{2pn(q^2-1)}\big(q^{2p-1}+q^{p+1}-2q^{p-1}-q^2-q+2\big)+\frac{q(p-1)^2}{2pn}+\frac{2(n-1)}{3pn}(2^{2p-3}+2^{p-2}-1).
\end{eqnarray*}

Based on the above discussions and Theorem \ref{corollary},  we obtain the following   result.
\begin{Corollary}
Let $n\geq5$ be a prime number.
Assume that $r=2p$, where $p\geq 3$ is a prime number satisfying $p\mid (q-1)$.
Then the number of extended irreducible
binary Goppa codes of length $2^n+1$ is at most
$$\frac{1}{2pn(q^2-1)}\Big(q^{2p-1}+q^{p+1}-2q^{p-1}-q^2-q+2\Big)+\frac{q(p-1)^2}{2pn}+\frac{2(n-1)}{3pn}(2^{2p-3}+2^{p-2}-1).$$
\end{Corollary}
$\bullet$ Subcase 2: $p\mid(q+1)$.
In this subcase we   have $p\nmid(q-1)$. First,
$$
\mathcal{N}_0=\frac{1}{2p}\sum_{d|2p}\mu(d)q^{\frac{2p}{d}}=\frac{1}{2p}(q^{2p}-q^p-q^2+q),~~
\mathcal{N}_1=\frac{q^2-1}{2p}\sum_{d|p \atop \gcd(2,d)=1}\mu(d)q^{\frac{p}{d}}=\frac{q^2-1}{2p}(q^p-q),
$$
$$
\mathcal{N}_2=q(q+1)\cdot\sum_{D|\gcd(2p,q-1) \atop D\neq 1}\frac{\varphi^2(D)}{2p}\sum_{d|\frac{2p}{D} \atop \gcd(d, D)=1}\mu(d)\big(q^{\frac{2p}{Dd}-1}\big)=0
$$
and
\begin{eqnarray*}
\mathcal{N}_3&=&\frac{q(q-1)}{2}\cdot\sum_{D|\gcd(2p,q+1) \atop D\neq 1}\frac{\varphi^2(D)}{2p}\sum_{d|\frac{2p}{D} \atop \gcd(d, D)=1}\mu(d)\big(q^{\frac{2p}{Dd}}+(-1)^{\frac{2p}{Dd}+1}\big)\\
&=&\frac{q(q-1)}{2}\cdot\sum_{D|p \atop D\neq 1}\frac{\varphi^2(D)}{2p}\sum_{d|\frac{2p}{D} \atop \gcd(d, D)=1}\mu(d)\big(q^{\frac{2p}{Dd}}+(-1)^{\frac{2p}{Dd}+1}\big)\\
&=& \frac{q(q-1)}{2}\cdot\frac{\varphi^2(p)}{2p}\sum_{d|2 \atop \gcd(d, p)=1}\mu(d)\big(q^{\frac{2}{d}}+(-1)^{\frac{2}{d}+1}\big)\\
&=& \frac{q(q-1)}{2}\cdot\frac{(p-1)^2}{2p}\big(q^2-1-(q+1)\big)\\
&=& \frac{(p-1)^2}{4p}\cdot q(q-1)(q^2-q-2)\\
&=& \frac{(p-1)^2}{4p}\cdot q(q-2)(q^2-1).
\end{eqnarray*}
Thus the number of orbits of ${\rm PGL}$ on   $\mathcal{I}_r$   is equal to
\begin{eqnarray*}
\big|{\rm PGL}\verb|\|\mathcal{I}_r\big|&=&\frac{1}{q(q^2-1)}\big(\mathcal{N}_0+\mathcal{N}_1+\mathcal{N}_2+\mathcal{N}_3\big)\\
&=&\frac{1}{q(q^2-1)}\Big(\frac{1}{2p}\big(q^{2p}-q^p-q^2+q+(q^2-1)(q^p-q)\big)+\frac{(p-1)^2}{4p}\cdot q(q-2)(q^2-1)\Big)\\
&=&\frac{1}{2pq(q^2-1)}\big(q^{2p}+q^{p+2}-2q^p-q^3-q^2+2q\big)+\frac{(q-2)(p-1)^2}{4p}\\
&=&\frac{1}{2p(q^2-1)}\big(q^{2p-1}+q^{p+1}-2q^{p-1}-q^2-q+2\big)+\frac{(q-2)(p-1)^2}{4p}.
\end{eqnarray*}
Next, we have
\begin{eqnarray*}
|\mathcal{X}|&=&\frac{1}{2p}\sum_{d|2p}\mu(d)(2^{\frac{2p}{d}}-1)=\frac{1}{2p}(2^{2p}-2^p-2^2+2)=\frac{1}{2p}(2^{2p}-2^p-2).\\
|\Delta_2|&=&\frac{1}{r}\sum_{d|\frac{r}{2} \atop d~odd}\mu(d)2^{\frac{r}{2d}}=\frac{1}{2p}\sum_{d|p \atop d~odd}\mu(d)2^{\frac{p}{d}}=\frac{1}{2p}(2^p-2).\\
|\Delta_5|&=&0.
\end{eqnarray*}
Thus the number $s_0$ of orbits of ${\rm P\Gamma L}$ on $\mathcal{I}_r$ with size $1$ is
$$s_0=\frac{1}{6}\big(|\mathcal{X}|+3|\Delta_2|+2|\Delta_5|\big)=\frac{1}{12p}(2^{2p}+2^{p+1}-8).$$
Let $s$ be the number of orbits of ${\rm P\Gamma L}$ on $\mathcal{I}_r$.
Then
$$s_0+n(s-s_0)=\big|{\rm PGL}\verb|\|\mathcal{I}_r\big|,$$
which leads to
\begin{eqnarray*}
&&\frac{1}{12p}(2^{2p}+2^{p+1}-8)+n(s-s_0)=\frac{1}{2p(q^2-1)}\big(q^{2p-1}+q^{p+1}-2q^{p-1}-q^2-q+2\big)+\frac{(q-2)(p-1)^2}{4p}\\
&\Rightarrow& 2^{2p}+2^{p+1}-8+12pn(s-s_0)=\frac{6}{q^2-1}\big(q^{2p-1}+q^{p+1}-2q^{p-1}-q^2-q+2\big)+3(q-2)(p-1)^2\\
&\Rightarrow& 12pn(s-s_0)=\frac{6}{q^2-1}\big(q^{2p-1}+q^{p+1}-2q^{p-1}-q^2-q+2\big)+3(q-2)(p-1)^2-2^{2p}-2^{p+1}+8\\
&\Rightarrow& s-s_0=\frac{1}{2pn(q^2-1)}\big(q^{2p-1}+q^{p+1}-2q^{p-1}-q^2-q+2\big)+\frac{(q-2)(p-1)^2}{4pn}-\frac{2^{2p}+2^{p+1}-8}{12pn}\\
&\Rightarrow& s=\frac{1}{2pn(q^2-1)}\big(q^{2p-1}+q^{p+1}-2q^{p-1}-q^2-q+2\big)+\frac{(q-2)(p-1)^2}{4pn}-\frac{2^{2p}+2^{p+1}-8}{12pn}+s_0\\
&\Rightarrow& s=\frac{1}{2pn(q^2-1)}\big(q^{2p-1}+q^{p+1}-2q^{p-1}-q^2-q+2\big)+\frac{(q-2)(p-1)^2}{4pn}+\frac{1}{12p}(1-\frac{1}{n})(2^{2p}+2^{p+1}-8)\\
&\Rightarrow& s=\frac{1}{2pn(q^2-1)}\big(q^{2p-1}+q^{p+1}-2q^{p-1}-q^2-q+2\big)+\frac{(q-2)(p-1)^2}{4pn}+\frac{2(n-1)}{3pn}(2^{2p-3}+2^{p-2}-1).
\end{eqnarray*}

We have arrived at the following result.
\begin{Corollary}\label{corollary2}
Let $n\geq5$ be a prime number.
Assume that $r=2p$, where $p\geq 3$ is a prime number satisfying $p\mid(q+1)$.
Then the number of extended irreducible
binary Goppa codes of length $2^n+1$ is at most
$$\frac{1}{2pn(q^2-1)}\Big(q^{2p-1}+q^{p+1}-2q^{p-1}-q^2-q+2\Big)+\frac{(q-2)(p-1)^2}{4pn}+\frac{2(n-1)}{3pn}(2^{2p-3}+2^{p-2}-1).$$
\end{Corollary}

Taking $p=3$ in the above corollary, we immediately reobtain the main result of \cite{yueqin}, as given below.
\begin{Corollary}
{\rm(\cite[Theorem 4.8]{yueqin})} Let $n\geq5$ be a prime number.
The number of extended irreducible binary sextic Goppa codes of length $2^n+1$
is at most
$$\frac{2^{3n}+2^{2n}+3\cdot 2^n+12n-18}{6n}.$$
\end{Corollary}
\begin{proof}
Taking $p=3$ in Corollary \ref{corollary2}, one has $r=2p=6$. It follows from $3\mid(2^n+1)$  that $p\mid(q+1)$.
Using Corollary \ref{corollary2}, we have
\begin{eqnarray*}
s&=&\frac{1}{2pn(q^2-1)}\big(q^{2p-1}+q^{p+1}-2q^{p-1}-q^2-q+2\big)+\frac{(q-2)(p-1)^2}{4pn}+\frac{2(n-1)}{3pn}(2^{2p-3}+2^{p-2}-1)\\
&=&\frac{1}{6n(q^2-1)}\big(q^{5}+q^{4}-2q^{2}-q^2-q+2\big)+\frac{q-2}{3n}+\frac{2n-2}{9n}(2^{3}+2-1)\\
&=&\frac{1}{6n(q^2-1)}\big(q^{5}+q^{4}-3q^2-q+2\big)+\frac{q-8}{3n}+2\\
&=&\frac{1}{6n}\big(q^{3}+q^{2}+q-2\big)+\frac{q-8}{3n}+2\\
&=&\frac{1}{6n}\big(q^{3}+q^{2}+3q-18\big)+2\\
&=&\frac{2^{3n}+2^{2n}+3\cdot 2^n+12n-18}{6n}.
\end{eqnarray*}
We are done.
\end{proof}
$\bullet$ Subcase 3: $p\nmid(q+1)$ and $p\nmid(q-1)$.
In this subcase, $\mathcal{N}_2=\mathcal{N}_3=0$,
$$
\mathcal{N}_0=\frac{1}{2p}\sum_{d|2p}\mu(d)q^{\frac{2p}{d}}=\frac{1}{2p}(q^{2p}-q^p-q^2+q),~~
\mathcal{N}_1=\frac{q^2-1}{2p}\sum_{d|p \atop \gcd(2,d)=1}\mu(d)q^{\frac{p}{d}}=\frac{q^2-1}{2p}(q^p-q).
$$
Thus the number of orbits of ${\rm PGL}$ on   $\mathcal{I}_r$  is equal to
\begin{eqnarray*}
\big|{\rm PGL}\verb|\|\mathcal{I}_r\big|&=&\frac{1}{q(q^2-1)}\big(\mathcal{N}_0+\mathcal{N}_1+\mathcal{N}_2+\mathcal{N}_3\big)\\
&=&\frac{1}{q(q^2-1)}\cdot\frac{1}{2p}\big(q^{2p}-q^p-q^2+q+(q^2-1)(q^p-q)\big)\\
&=&\frac{1}{2pq(q^2-1)}\cdot \big(q^{2p}+q^{p+2}-2q^p-q^3-q^2+2q\big)\\
&=&\frac{1}{2p(q^2-1)}\cdot \big(q^{2p-1}+q^{p+1}-2q^{p-1}-q^2-q+2\big).
\end{eqnarray*}
Next, we have $|\Delta_5|=0$,
\begin{eqnarray*}
|\mathcal{X}|&=&\frac{1}{2p}\sum_{d|2p}\mu(d)(2^{\frac{2p}{d}}-1)=\frac{1}{2p}(2^{2p}-2^p-2^2+2)=\frac{1}{2p}(2^{2p}-2^p-2),\\
|\Delta_2|&=&\frac{1}{r}\sum_{d|\frac{r}{2} \atop d~odd}\mu(d)2^{\frac{r}{2d}}=\frac{1}{2p}\sum_{d|p \atop d~odd}\mu(d)2^{\frac{p}{d}}=\frac{1}{2p}(2^p-2).\\
\end{eqnarray*}
Thus the number $s_0$ of orbits of ${\rm P\Gamma L}$ on $\mathcal{I}_r$ with size $1$ is
$$s_0=\frac{1}{6}\big(|\mathcal{X}|+3|\Delta_2|+2|\Delta_5|\big)=\frac{1}{12p}(2^{2p}+2^{p+1}-8).$$
Let $s$ be the number of orbits of ${\rm P\Gamma L}$ on $\mathcal{I}_r$.
Then
$$s_0+n(s-s_0)=\big|{\rm PGL}\verb|\|\mathcal{I}_r\big|,$$
and thus,
\begin{eqnarray*}
&&\frac{1}{12p}(2^{2p}+2^{p+1}-8)+n(s-s_0)=\frac{1}{2p(q^2-1)}\big(q^{2p-1}+q^{p+1}-2q^{p-1}-q^2-q+2\big)\\
&\Rightarrow& 2^{2p}+2^{p+1}-8+12pn(s-s_0)=\frac{6}{q^2-1}\big(q^{2p-1}+q^{p+1}-2q^{p-1}-q^2-q+2\big)\\
&\Rightarrow& 12pn(s-s_0)=\frac{6}{q^2-1}\big(q^{2p-1}+q^{p+1}-2q^{p-1}-q^2-q+2\big)-2^{2p}-2^{p+1}+8\\
&\Rightarrow& s-s_0=\frac{1}{2pn(q^2-1)}\big(q^{2p-1}+q^{p+1}-2q^{p-1}-q^2-q+2\big)-\frac{2^{2p}+2^{p+1}-8}{12pn}\\
&\Rightarrow& s=\frac{1}{2pn(q^2-1)}\big(q^{2p-1}+q^{p+1}-2q^{p-1}-q^2-q+2\big)-\frac{2^{2p}+2^{p+1}-8}{12pn}+s_0\\
&\Rightarrow& s=\frac{1}{2pn(q^2-1)}\big(q^{2p-1}+q^{p+1}-2q^{p-1}-q^2-q+2\big)+\frac{2(n-1)}{3pn}(2^{2p-3}+2^{p-2}-1).
\end{eqnarray*}

We have obtained the following result.
\begin{Theorem}
Let $n\geq5$ be a prime number.
Assume that $r=2p$, where $p\geq 3$ is a prime number satisfying $p\nmid(q+1)$ and $p\nmid(q-1)$.
Then the number of extended irreducible
binary Goppa codes of length $2^n+1$ is at most
$$\frac{1}{2pn(q^2-1)}\Big(q^{2p-1}+q^{p+1}-2q^{p-1}-q^2-q+2\Big)+\frac{2(n-1)}{3pn}(2^{2p-3}+2^{p-2}-1).$$
\end{Theorem}

\subsection{The case: $\gcd(r, 2(q^2-1))=1$.}
In this case we have   $2\nmid r$, $3\nmid r$, $\gcd(r, q-1)=1$ and $\gcd(r, q+1)=1$. After simple calculations,
we have
$$
\mathcal{N}_0 = \frac{1}{r}\sum_{d|r}\mu(d)q^{\frac{r}{d}},~~
\mathcal{N}_1 =\mathcal{N}_2=\mathcal{N}_3=0.
$$
Additionally, we have $|\Delta_2|=|\Delta_5|=0$ and
$$
\big{|}\mathcal{X}\big|
=\frac{1}{r}\sum_{d|r}\big(2^{\frac{r}{d}}-1\big)\mu(d).
$$
We can give an upper bound for the
number of inequivalent  extended irreducible binary Goppa codes of length $2^n+1$ and degree $r$ with $\gcd(r,2(q^2-1))=1$, which is the main
result of \cite{cz}.
\begin{Corollary}
{\rm (\cite[Theorem 3.11]{cz})} We assume that $n\geq 5$ is an odd prime number,
$q=2^n$,  and $r\geq3$ is a positive integer satisfying $\gcd(r,n)=1$ and $\gcd(r,q^3-q)=1$.
The number of inequivalent extended irreducible binary Goppa codes of length $q+1$ and degree $r$ is at most
$$\frac{n-1}{6rn}\cdot\sum_{d|r}\big(2^{\frac{r}{d}}-1\big)\mu(d)
+\frac{1}{rnq(q^2-1)}\cdot\sum_{d|r}\mu(d)q^{\frac{r}{d}},$$
where $\mu$ is the M\"{o}bius function.
\end{Corollary}

%\vspace{0.3 cm}
%\noindent{\bf Acknowledgements}\quad
%We are very grateful to
%Professor Qin Yue in Nanjing University of Aeronautics and Astronautics for
%introducing us to the topic of enumeration of Goppa codes, and sending their manuscript \cite{yueqin} to us.

\noindent {\bf \large Appendix}
%\section{Appendix}

\medskip

\textbf{Proof of Lemma \ref{orbitnumber}:}
First, there are four families of conjugacy classes of the general linear group ${\rm GL}$ whose representatives are given as follows, see \cite[pages 324-326]{jl}.

(i) The matrices
$$sE_2=\begin{pmatrix}
s & 0 \\
0 & s
\end{pmatrix}(s\in \mathbb{F}_q^*)
$$
belong to the centre of ${\rm GL}$. They give  $q-1$ conjugacy classes of ${\rm GL}$ with size $1$.

(ii) Consider the matrices
$$U_s=\begin{pmatrix}
s & 1 \\
0 & s
\end{pmatrix}(s\in \mathbb{F}_q^*).
$$
The matrices $U_s~(s\in \mathbb{F}_q^*)$ give  $q-1$ conjugacy classes of ${\rm GL}$.
Each conjugacy class contains $q^2-1$ elements.

(iii) Let
$$D_{s,t}=\begin{pmatrix}
s & 0 \\
0 & t
\end{pmatrix}(s, t\in \mathbb{F}_q^*, ~s\neq t).
$$
The matrices $D_{s,t}~(s, t\in \mathbb{F}_q^*, ~s\neq t)$ give  $\frac{(q-1)(q-2)}{2}$ conjugacy classes of ${\rm GL}$
(note that $D_{s,t}$ and $D_{t,s}$ belong to the same conjugacy class).
Each conjugacy class contains $q(q+1)$ elements.

(iv) Consider
$$V_{\gamma}=\begin{pmatrix}
0 & 1 \\
\gamma^{1+q} & \gamma+\gamma^q
\end{pmatrix}\Big(\gamma\in \bigcup_{i=1}^{\frac{q}{2}}\xi^{(q-1)i}\mathbb{F}_q^*\Big).
$$
The matrices $V_{\gamma}~(\gamma\in \bigcup_{i=1}^{\frac{q}{2}}\xi^{(q-1)i}\mathbb{F}_q^*)$ give  $\frac{q(q-1)}{2}$
conjugacy classes of ${\rm GL}$.
Each conjugacy class contains $q(q-1)$ elements.

By this  result, we can determine the   conjugacy classes of   ${\rm PGL}$.
The   representatives of the conjugacy classes of ${\rm PGL}$ are divided into four cases, as we listed below.

(1) $E_2$.

(2) Note that for each $s\in \mathbb{F}_q^*$,
$$s\begin{pmatrix}
1 & 0 \\
0 & s
\end{pmatrix}
\begin{pmatrix}
1 & 1 \\
0 & 1
\end{pmatrix}
=\begin{pmatrix}
s & 1 \\
0 & s
\end{pmatrix}
\begin{pmatrix}
1 & 0 \\
0 & s
\end{pmatrix}.
$$
For any $s\in \mathbb{F}_q^*$ it follows that
$$\begin{pmatrix}
s & 1 \\
0 & s
\end{pmatrix}\sim
\begin{pmatrix}
1 & 1 \\
0 & 1
\end{pmatrix}~ \mbox{in PGL}.$$
Thus the elements
$$U_s=\begin{pmatrix}
s & 1 \\
0 & s
\end{pmatrix}(s\in \mathbb{F}_q^*)$$
provide a conjugacy class with representative $$\begin{pmatrix}
1 & 1 \\
0 & 1
\end{pmatrix}$$ of the group ${\rm PGL}$.
Clearly, $\begin{pmatrix}
1 & 1 \\
0 & 1
\end{pmatrix}$ and $E_2$ do not belong to the same conjucagy class.

(3) Note that the elements
$$D_{s,t}=\begin{pmatrix}
s & 0 \\
0 & t
\end{pmatrix}(s, t\in \mathbb{F}_q^*, ~s\neq t)$$
give  $q-2$ elements of ${\rm PGL}$ as follows:
$$\begin{pmatrix}
1 & 0 \\
0 & a
\end{pmatrix}(1\neq a\in \mathbb{F}_q^*).$$
Additionally, for $1\neq a\in \mathbb{F}_q^*$,
$$
\begin{pmatrix}
a & 0 \\
0 & 1
\end{pmatrix}
\begin{pmatrix}
1 & 0 \\
0 & a^{-1}
\end{pmatrix}
\begin{pmatrix}
a^{-1} & 0 \\
0 & 1
\end{pmatrix}=
\begin{pmatrix}
1 & 0 \\
0 & a^{-1}
\end{pmatrix}=
a^{-1}
\begin{pmatrix}
a & 0 \\
0 & 1
\end{pmatrix}.
$$
This shows that
$$\begin{pmatrix}
1 & 0 \\
0 & a^{-1}
\end{pmatrix}\sim
\begin{pmatrix}
a & 0 \\
0 & 1
\end{pmatrix}~\mbox{in PGL}.
$$
Clearly,
$$\begin{pmatrix}0 & 1 \\1 & 0\end{pmatrix}^{-1}\begin{pmatrix}1 & 0 \\0 & a\end{pmatrix}
\begin{pmatrix}0 & 1 \\1 & 0\end{pmatrix}=\begin{pmatrix}a & 0 \\0 & 1\end{pmatrix},$$
which implies that
$$\begin{pmatrix}
a & 0 \\
0 & 1
\end{pmatrix}\sim
\begin{pmatrix}
1 & 0 \\
0 & a
\end{pmatrix}~\mbox{in PGL}.
$$
Hence,
$$\begin{pmatrix}
1 & 0 \\
0 & a
\end{pmatrix}\sim
\begin{pmatrix}
1 & 0 \\
0 & a^{-1}
\end{pmatrix}~\mbox{in PGL}$$

Let $S$ be a subset of $\mathbb{F}_q^*$ such that $\{1\}\cup S\cup S^{-1}=\mathbb{F}_q^*$, where $S^{-1}=\{s^{-1}\,|\,s\in S\}$.
In the following we prove that for any $a,b\in S$ and $a\neq b$,
$\begin{pmatrix}
1 & 0 \\
0 & a
\end{pmatrix}$ and
$\begin{pmatrix}
1 & 0 \\
0 & b
\end{pmatrix}$ do not belong to the same conjugacy class.
Suppose that $\begin{pmatrix}
1 & 0 \\
0 & a
\end{pmatrix}$ and
$\begin{pmatrix}
1 & 0 \\
0 & b
\end{pmatrix}$ belong to the same conjugacy class,
then there exists $\lambda \in \mathbb{F}_q^*, P\in {\rm GL}$ such that
$$\lambda \begin{pmatrix}
1 & 0 \\
0 & a
\end{pmatrix}=P\begin{pmatrix}
1 & 0 \\
0 & b
\end{pmatrix}P^{-1}.$$
Since the conjugate matrices have the same eigenvalues, we have that
$$\{\lambda, \lambda a\}=\{1, b\}.$$
If $\lambda=1$, then $a=b$; if $\lambda\neq 1$, then $\lambda=b, \lambda a=1$ and so $b=a^{-1}$.
In either case we can get a contradiction. Hence for any $a,b\in S$ and $a\neq b$,
$\begin{pmatrix}
1 & 0 \\
0 & a
\end{pmatrix}$ and
$\begin{pmatrix}
1 & 0 \\
0 & b
\end{pmatrix}$ do not belong to the same conjugacy class.
Thus these $q-2$ elements
$$\begin{pmatrix}
1 & 0 \\
0 & a
\end{pmatrix}(1\neq a\in \mathbb{F}_q^*)$$
of ${\rm PGL}$ provide  $\frac{q-2}{2}$ conjugacy classes with representatives
$$\begin{pmatrix}
1 & 0 \\
0 & a
\end{pmatrix}(a\in S),$$
where $S\subseteq \mathbb{F}_q^*$ satisfying $\{1\}\cup S\cup S^{-1}=\mathbb{F}_q^*$ and $S^{-1}=\{s^{-1}|s\in S\}$.

In the following it remains to show that $\begin{pmatrix}
1 & 0 \\0 & a\end{pmatrix}(a\in S)$
cannot   conjugate  to $E_2$ and
$\begin{pmatrix}1 & 1 \\0 & 1\end{pmatrix}$, respectively. Since $a\neq 1$, $\begin{pmatrix}1 & 0 \\0 & a\end{pmatrix}(a\in S)$
cannot  conjugate to $E_2$.
Suppose that $\begin{pmatrix}1 & 0 \\0 & a\end{pmatrix}(a\in S)$
 conjugates to
$\begin{pmatrix}1 & 1 \\0 & 1\end{pmatrix}$.
Then there exists $\lambda \in \mathbb{F}_q^*, P\in {\rm GL}$ such that
$$\lambda\begin{pmatrix}1 & 0 \\0 & a\end{pmatrix}=P\begin{pmatrix}1 & 1 \\0 & 1\end{pmatrix}P^{-1}.$$
So $\lambda=\lambda a=1$, which gets $a=1$. This is a contradiction. Hence
$\begin{pmatrix}1 & 0 \\0 & a\end{pmatrix}(a\in S)$
do not conjugate to
$\begin{pmatrix}1 & 1 \\0 & 1\end{pmatrix}$.

(4) First we are going to prove that for any $\lambda\in \mathbb{F}_q^*$,
$$
V_{\gamma}=\begin{pmatrix}
0 & 1 \\
\gamma^{1+q} & \gamma+\gamma^q
\end{pmatrix}\sim
V_{\lambda\gamma}=\begin{pmatrix}
0 & 1 \\
\lambda^2\gamma^{1+q} & \lambda\gamma+\lambda\gamma^q
\end{pmatrix}~\mbox{in PGL}.
$$

Note that
$$\lambda V_{\gamma}=\begin{pmatrix}
0 & \lambda \\
\lambda\gamma^{1+q} & \lambda\gamma+\lambda\gamma^q
\end{pmatrix}\sim
\begin{pmatrix}
\lambda\gamma & 0 \\
0 & \lambda\gamma^q
\end{pmatrix}~ \mbox{in}~{\rm GL}_2(\mathbb{F}_{q^2}),
$$
$$V_{\lambda\gamma}=\begin{pmatrix}
0 & 1 \\
\lambda^2\gamma^{1+q} & \lambda\gamma+\lambda\gamma^q
\end{pmatrix}\sim
\begin{pmatrix}
\lambda\gamma & 0 \\
0 & \lambda\gamma^q
\end{pmatrix}~ \mbox{in}~{\rm GL}_2(\mathbb{F}_{q^2}).$$
Then,
$$
\lambda V_{\gamma}=\begin{pmatrix}
0 & \lambda \\
\lambda\gamma^{1+q} & \lambda\gamma+\lambda\gamma^q
\end{pmatrix}\sim
V_{\lambda\gamma}=\begin{pmatrix}
0 & 1 \\
\lambda^2\gamma^{1+q} & \lambda\gamma+\lambda\gamma^q
\end{pmatrix}
~ \mbox{in}~{\rm GL}_2(\mathbb{F}_{q^2}).$$
We have
$$
\lambda V_{\gamma}=\begin{pmatrix}
0 & \lambda \\
\lambda\gamma^{1+q} & \lambda\gamma+\lambda\gamma^q
\end{pmatrix}\sim
V_{\lambda\gamma}=\begin{pmatrix}
0 & 1 \\
\lambda^2\gamma^{1+q} & \lambda\gamma+\lambda\gamma^q
\end{pmatrix}
~ \mbox{in}~{\rm GL}_2(\mathbb{F}_{q}).$$
Hence, for any $\lambda\in \mathbb{F}_q^*$,
$$
V_{\gamma}=\begin{pmatrix}
0 & 1 \\
\gamma^{1+q} & \gamma+\gamma^q
\end{pmatrix}\sim
V_{\lambda\gamma}=\begin{pmatrix}
0 & 1 \\
\lambda^2\gamma^{1+q} & \lambda\gamma+\lambda\gamma^q
\end{pmatrix}~ \mbox{in}~{\rm PGL}.
$$

Let $\gamma_1=\xi^{(q-1)i_1}, \gamma_2=\xi^{(q-1)i_2}$, where $1\leq i_1, i_2\leq \frac{q}{2}$ and $i_1\neq i_2$.
Secondly, we prove that $V_{\gamma_1}$ does not conjugate to $V_{\gamma_2}$.
Suppose otherwise that $V_{\gamma_1}$   conjugates to $V_{\gamma_2}$.
Then there exists $\lambda_0 \in \mathbb{F}_q^*, P\in {\rm GL}$ such that
$$\lambda_0V_{\gamma_1}=PV_{\gamma_2}P^{-1},$$
i.e.,
$$\lambda_0\begin{pmatrix}0 & 1 \\ \gamma_1^{1+q} & \gamma_1+\gamma_1^q\end{pmatrix}
=P\begin{pmatrix}0 & 1 \\ \gamma_2^{1+q} & \gamma_2+\gamma_2^q\end{pmatrix}P^{-1},$$
which implies that
$$\big\{\lambda_0\gamma_1, \lambda_0\gamma_1^q\big\}=\big\{\gamma_2, \gamma_2^q\big\}.$$
Note that if $\lambda_0\gamma_1=\gamma_2$, then
\begin{eqnarray*}
\lambda_0=\frac{\gamma_2}{\gamma_1}
&\Rightarrow & \lambda_0=\xi^{(q-1)(i_2-i_1)}\\
&\Rightarrow & \xi^{(q-1)^2(i_2-i_1)}=1\\
&\Rightarrow & (q-1)^2(i_2-i_1)\equiv 0 ~\big({\rm mod}(q^2-1)\big)\\
&\Rightarrow & (q-1)(i_2-i_1)\equiv 0 ~\big({\rm mod}(q+1)\big)\\
&\Rightarrow & i_2-i_1 \equiv 0 ~\big({\rm mod}(q+1)\big).
\end{eqnarray*}
Since $1\leq i_1, i_2\leq \frac{q}{2}$ and $i_1\neq i_2$, this is a contradiction.
In addition, if $\lambda_0\gamma_1=\gamma_2^q$, then
\begin{eqnarray*}
\lambda_0\xi^{(q-1)i_1}=\xi^{(q-1)i_2}&\Rightarrow & \lambda_0=\xi^{(q-1)(qi_2-i_1)}\\
&\Rightarrow & \xi^{(q-1)^2(qi_2-i_1)}=1\\
&\Rightarrow & \xi^{(q-1)^2(i_2-i_1)}=1\\
&\Rightarrow & (q-1)^2(qi_2-i_1)\equiv 0 ~\big({\rm mod}(q^2-1)\big)\\
&\Rightarrow & (q-1)(qi_2-i_1)\equiv 0 ~\big({\rm mod}(q+1)\big)\\
&\Rightarrow & i_2+i_1 \equiv 0 ~\big({\rm mod}(q+1)\big).
\end{eqnarray*}
The same reason shows that this is also a contradiction.

Therefore the matrices
$$V_{\gamma_i}=\begin{pmatrix}
0 & 1 \\
\gamma_i^{1+q} & \gamma_i+\gamma_i^q
\end{pmatrix}$$
give  $\frac{q}{2}$ conjugacy classes,
where $\gamma_i=\xi^{(q-1)i}, i=1,2,\cdots,\frac{q}{2}$.

Using the same arguments as above,  it follows that
$$V_{\gamma_i}=\begin{pmatrix}
0 & 1 \\
\gamma_i^{1+q} & \gamma_i+\gamma_i^q
\end{pmatrix}\Big(\gamma_i=\xi^{(q-1)i}, ~i=1,2,\cdots,\frac{q}{2}\Big)$$
do not conjugate to $E_2$, $\begin{pmatrix}1 & 1 \\0 & 1\end{pmatrix}$ and
$\begin{pmatrix}1 & 0 \\0 & a\end{pmatrix}(a\in S)$, respectively.

Lastly, the conjugacy classes we have found account for
$$1+(q^2-1)+\frac{q-2}{2}\cdot q(q+1)+\frac{q}{2}\cdot q(q-1)$$
elements altogether. This sum is equal to the order of the  group ${\rm PGL}$,
so we have found all the conjugacy classes. We are done.

\end{document}